\documentclass[a4paper,11pt]{article}
\usepackage{amsmath,amssymb,amsfonts,amsthm,stmaryrd}
\usepackage{color}
\usepackage{geometry}
\usepackage{tikz}
\usepackage{graphicx}
\usepackage{url}
\usepackage{hyperref}
\hypersetup{
    colorlinks=true,
    citecolor=blue,
    linkcolor=red,
    filecolor=magenta,
    urlcolor=cyan,
}
\let\set\mathbb

\newtheorem{theorem}{Theorem}[section]
\newtheorem{corollary}[theorem]{Corollary}
\newtheorem{lemma}[theorem]{Lemma}
\newtheorem{remark}[theorem]{Remark}
\newtheorem{definition}[theorem]{Definition}
\newtheorem{example}[theorem]{Example}
\newtheorem{proposition}[theorem]{Proposition}

\def\hdeg{\operatorname{hdeg}}
\def\tdeg{\operatorname{tdeg}}
\def\redt{\operatorname{red}}
\def\im{\operatorname{im}}
\def\spanning{\operatorname{span}}

\DeclareMathOperator{\gl}{GL}
\newcommand{\qbinom}{\genfrac{[}{]}{0pt}{}}

\title{A Unified Reduction for Hypergeometric and $q$-Hypergeometric Creative Telescoping\footnote{S.\ Chen was partially 
supported by the National Key R\&D Programs of
China (No.\ 2020YFA0712300 and No.\ 2023YFA1009401), the NSFC grants (No.\ 12271511 and No.\ 11688101), 
the CAS Funds of the Youth Innovation Promotion Association (No.\ Y2022001), and 
the Strategic Priority Research Program of the Chinese Academy of Sciences (No.\ XDB0510201).
H.\ Du was supported by the NSFC grant (No.\ 12201065).
Y.\ Gao was supported by the Austrian Science Fund (FWF) grant 10.55776/PAT1332123.
H.\ Huang was partially supported by the NSFC grant (No.\ 12101105) and the
Natural Science Foundation of Fujian Province of China (No.\ 2024J01271).
Z.\ Li was partially supported by the NSFC grant (No.\ 12271511)
and the National Key R\&D Program of China (No.\ 2023YFA1009401).
}}
\author{Shaoshi Chen$^{1,2}$, Hao Du$^3$, Yiman Gao$^{ 4}$, Hui Huang$^ {5}$, Ziming Li$^{1,2}$\\
	{\small $^1$ KLMM, Academy of Mathematics and Systems Science, Chinese Academy of Sciences,}\\
	{\small Beijing 100190, China}\\
	{\small $^2$ School of Mathematical Sciences, University of Chinese Academy of Sciences,}\\
	{\small Beijing 100049, China}\\
	{\small $^3$ School of Mathematical Sciences, Beijing University of Posts and Telecommunications,}\\
	{\small Beijing 100876, China}\\
    {\small  $^4$ Johannes Kepler University Linz, Research Institute for Symbolic Computation(RISC),}\\
	{\small Altenberger Stra\ss e 69, A-4040, Linz, Austria}\\
	{\small  $^5$ School of Mathematics and Statistics, Fuzhou University,}\\
	{\small Fuzhou 350108, China}\\
	{\tt \small schen@amss.ac.cn, haodu@bupt.edu.cn,  ymgao@risc.jku.at}\\
	{\tt \small huanghui@fzu.edu.cn, zmli@mmrc.iss.ac.cn}
	}
\date{}

\begin{document}
\maketitle
{\em Dedicated to Professors George Andrews and Bruce Berndt for their 85th birthdays}

\begin{abstract}
We adapt the theory of normal and special polynomials from symbolic integration
to the summation setting, and then built up a general framework embracing both the 
usual shift case and the $q$-shift case. In the context of this general framework, we develop
a unified reduction algorithm, and subsequently a creative telescoping algorithm,
applicable to both hypergeometric terms and their $q$-analogues.
Our algorithms allow to split up the usual shift case and the $q$-shift case only when it is really 
necessary, and thus instantly reveal the intrinsic differences between these two cases. 
Computational experiments are also provided.
\end{abstract}

\section{Introduction}\label{SEC:intro}
Hypergeometric summation and its $q$-analogue appear frequently in combinatorics~\cite{AnAs1999}.
These are sums whose summands are ($q$-)hypergeometric terms, typically involving 
rational functions, geometric terms, factorial terms, binomial coefficients, and so on. 
Given a ($q$-)hypergeometric sum, an important problem is to decide whether the sum admits 
a \lq\lq closed form\rq\rq. A prominent technique for tackling such a problem is the method of
{\em creative telescoping}, also known as {\em Zeilberger's algorithm} in the hypergeometric
case and {\em $q$-Zeilberger's algorithm} in the $q$-hypergeometric case. 
This method was first pioneered by 
Zeilberger~\cite{Zeil1990a,Zeil1990b,Zeil1991} in the 1990s and has now become the 
primary technique for definite summation and integration.

In the case of summation, the method of creative telescoping takes a summand $f(x,y)$
and looks for polynomials $c_0,c_1,\dots,c_\rho$ in $x$ only, not all zero, and another term 
$g(x,y)$ in the same class as $f(x,y)$ such that 
\begin{equation}\label{EQ:ct}
c_0(x)f+c_1(x)S_x(f)+\dots+c_\rho(x)S_x^\rho(f) = S_y(g)-g,
\end{equation}
where $S_x$ and $S_y$ denote ($q$-)shift operators with respect to $x$ and $y$, respectively.
The number~$\rho$ may or may not be part of the input. If such $c_0,c_1,\dots,c_\rho$ and $g$ exist,
then the nonzero recurrence operator $L = c_0+c_1S_x+\cdots+c_\rho S_x^\rho$
is called a {\em telescoper} for $f$ and the term $g$ is called the {\em certificate} for~$L$.
From the relation~\eqref{EQ:ct}, one can derive a recurrence equation admitting
the given definite sum as a solution. With this equation at hand, we are then able to evaluate
the sum by some other available algorithms (for example, cf.\ \cite{Petk1992,APP1998}) or (dis-)prove 
an already given identity by substitution and initial-value checking. As an example, let us
consider the $q$-Chu-Vandermonde identity in the form
\[
\sum_{k=0}^n\qbinom{n}{k}_q\qbinom{b}{k}_q q^{k^2} = \qbinom{b+n}{n}_q,
\]
where 
\[
\qbinom{n}{k}_q=\begin{cases}
\frac{(q;q)_n}{(q;q)_k(q;q)_{n-k}} & \text{if}\ 0\leq k\leq n,\\[1ex]
0 & \text{otherwise}
\end{cases}
\]
 is the Gaussian binomial coefficient
and $(q;q)_k=\prod_{i=1}^k(1-q^i)$ is the $q$-Pochhammer symbol with $(q;q)_0=1$. Let $f_{n,k}$ denote the summand 
on the left-hand side of the identity, and let $x=q^n$ and $y=q^k$.
The method of creative telescoping then constructs a telescoper 
$L = q(1-q^{b+1}x)- q(1-qx)S_x$ for~$f_{n,k}$ and a corresponding certificate 
\[
g_{n,k} = \frac{q^2x(y-1)^2}{-qx+y}\qbinom{n}{k}_q\qbinom{b}{k}_q q^{k^2},
\]
where $S_x(f_{n,k}) = f_{n+1,k}$. Thus \eqref{EQ:ct} becomes
\[
q(1-q^{b+1}x)f_{n,k}- q(1-qx)f_{n+1,k} = g_{n,k+1}-g_{n,k}.
\]
Summing over $k$ from zero to $n$ on both sides, along with a subsequent simplification, delivers 
the recurrence equation 
\begin{equation}\label{EQ:CVrec}
q(1-q^{b+1}x)F_n- q(1-qx)F_{n+1} = 0\quad \text{with}\ F_n = \sum_{k=0}^nf_{n,k}.
\end{equation}
Using a $q$-analogue of Pascal's formula (cf.\ \cite[Exercises 1.2~(v)]{GaRa2004}), we see that the right-hand 
side of the identity $\qbinom{b+n}{n}_q$ satisfies the same recurrence equation. The correctness of the 
$q$-Chu-Vandermonde identity is finally confirmed by checking the equality of initial values at $n = 0$.
Similar reasoning processes apply to most of the summation identities listed in \cite[Appendix~II]{GaRa2004}.

Over the past 35 years, a variety of generalizations and improvements of creative telescoping
have been developed. As outlined in the introduction of~\cite{CHKL2015}, we can distinguish four 
generations among them. The first generation was based on elimination techniques. The second 
generation starts with Zeilberger's algorithm, and uses the idea of parametrizing an algorithm for
indefinite summation (or integration). The third generation was initiated by Apagodu and 
Zeilberger, and mainly applies a second-generation algorithm
by hand to a generic input so as to reduce the problem to solving linear systems.
Due to the efficiency and practicability, the algorithms of the second generation have been
implemented in many computer algebra systems, including {\sc Maple} and  {\sc Mathematica}, 
and are widely used in proving identities from combinatorics, the theory of partitions or physics, etc. 
More details can be found in~\cite{PWZ1996} for the first two generations and in~\cite{MoZe2005,ApZe2006} 
for the third one.

In terms of the fourth generation, 
reduction methods currently provide the state of the art for constructing telescopers
(see \cite{Chen2019} and the references therein). These convert a given summand $f(x,y)$
into some sort of reduced form $\redt(f)$ modulo all terms that are $y$-differences of other terms. 
The object is then to continually take shifts 
$f,S_x(f),S_x^2(f),\dots$ and, by reduction methods, convert these to reduced forms 
$\redt(f), \redt(S_x(f)),\redt(S_x^2(f)), \dots$ until we find a set of polynomials 
$c_0(x),c_1(x),\dots,c_\rho(x)$, not all zero, such that 
\[
c_0(x)\redt(f)+c_1(x)\redt(S_x(f))+\dots+c_\rho(x)\redt(S_x^\rho(f))=0.
\] 
This will give rise to a telescoper $L = c_0+c_1S_x+\cdots +c_\rho S_x^\rho$ for~$f$. 
The termination of the above process is guaranteed by known existence criteria for
telescopers (for example, see \cite{Abra2003} for the hypergeometric case and \cite{CHM2005} 
its $q$-analogue).

Compared with previous generations, the key advantage of the fourth generation is that it separates 
the computation of the $c_i$ from the computation of the $g$ in~\eqref{EQ:ct}, and thus enables one to find 
a telescoper without also necessarily computing a corresponding certificate. This is desirable 
in a typical situation where only the telescoper is of interest and its size is much smaller than 
the size of the certificate. 
For instance, in the above example of the $q$-Chu-Vandermonde identity, due to the natural 
boundary of Gaussian binomial coefficients, we could have directly used the telescoper to obtain 
the recurrence equation \eqref{EQ:CVrec} without knowing the certificate.
So far this approach has been worked out for various 
classes of functions, including rational functions \cite{BCCL2010,BLS2013,CHHLW2022}, 
hyperexponential functions \cite{BCCLX2013}, algebraic functions \cite{CKK2016}, 
D-finite functions \cite{CvHKK2018,BCLS2018,vdHo2021,CDK2023}, 
hypergeometric terms \cite{CHKL2015,Huan2016}, P-recursive sequences~\cite{vdHo2018,BrSa2024,CDKW2025}, and so on. 
One goal of the present paper is to further enlarge the list and include another 
important class of $q$-hypergeometric terms.

$q$-Hypergeometric terms are just slight adaptations of the usual ones by essentially promoting
involved variables to exponents of an additional parameter~$q$. One of the reasons for interest
in $q$-analogues is that, due to the extra parameter $q$, they have many counting interpretations
which are useful in combinatorics and analysis; see the classic books~\cite{Andr1976,Andr1986}
for many interesting applications in combinatorics, analysis and elsewhere in mathematics.
Very often, techniques for handling 
the usual case carry over to the $q$-analogue with some subtle modifications (cf.\ \cite{Koor1993,PaRi1997}).
Rather than working out these modifications individually, we aim to set up a general framework
which combines both the usual shift case and the $q$-shift case so as to reveal more profound reasons 
for this phenomenon. The foundation of this framework is the theory of normal and special 
polynomials adapted from symbolic integration~\cite{Bron2005}. 
With this theory, every rational function can be uniquely decomposed as a sum of the polynomial,
normal, and special parts.
As indicated by Lemma~\ref{LEM:special}, a major difference from the $q$-shift case to the usual one 
lies in the appearance of nontrivial special polynomials (and thus possibly nontrivial special parts).
This results in the Laurent polynomial reduction, instead of the usual polynomial reduction, for 
$q$-hypergeometric terms in~\cite{DHL2018}.
In order to eliminate this discrepancy, we introduce the notion of standard rational functions
(see Definition~\ref{DEF:standard}), which enables us to transform a nontrivial special part
to a ``simpler" term that can be tackled simultaneously with the polynomial part.
In this way, we unify the reduction processes for hypergeometric terms and their $q$-analogues,
and subsequently obtain a unified creative telescoping algorithm.
This algorithm extends the one developed in~\cite{CHKL2015} for the usual hypergeometric terms
by including the $q$-shift case, and it shares the important feature of the reduction-based approach 
that the computation of a telescoper is separated from that of its certificate. 

The remainder of the paper proceeds as follows. 
Some basic notions and results are recalled in the next section. In particular,
we briefly translate the theory of normal and special polynomials from symbolic integration
to the summation setting. Using this theory, we present in Section~\ref{SEC:red} 
a reduction algorithm which brings every ($q$-)hypergeometric term to a reduced form.
These reduced forms are shown in Section~\ref{SEC:rem} to be well-behaved with respect to 
taking linear combinations. Based on the reduction algorithm developed previously, Section~\ref{SEC:rct} 
describes an algorithm for constructing telescopers for ($q$-)hypergeometric terms,
followed in Section~\ref{SEC:exp} by an experimental comparison between our algorithms and 
the built-in algorithms of Maple. 

\section{Preliminaries}\label{SEC:prelim}
Throughout the paper, let $\set F$ be a field of characteristic zero, and $\set F(y)$ be the field 
of rational functions in~$y$ over~$\set F$. Let $\sigma_y$ be an $\set F$-automorphism of $\set F(y)$.
The pair $(\set F(y),\sigma_y)$ is called a {\em difference field}. By~\cite[Theorem~6.2.3]{Wein2009}, 
there exists a matrix 
\[
A = \begin{pmatrix}
a & b\\
c & d
\end{pmatrix}
\in \gl_2(\set F)
\quad\text{with}\ 
a,b,c,d\in\set F\ \text{and}\ ad-bc \neq 0,
\]
such that $\sigma_y=\varphi_A$, where $\varphi_A$ denotes the $\set F$-automorphism of $\set F(y)$ 
defined by
\[
\varphi_A(f(y)) = f\Big(\frac{ay+b}{cy+d}\Big)
\quad\text{for all}\ f\in\set F(y).
\]
Let
\[
A_u= \begin{pmatrix} 1 & 1\\0 & 1\end{pmatrix}
\quad\text{and}\quad
A_q = \begin{pmatrix} q & 0 \\ 0 & 1\end{pmatrix} \ \text{with}\ q\in \set F\setminus\{0\}.
\]
We call $\sigma_y$ the {\em usual shift operator} if $\sigma_y = \varphi_{A_u}$, and call it the
{\em $q$-shift operator} if $\sigma_y=\varphi_{A_q}$.

According to the discussions in \cite[\S 2, Page 323]{ChSi2014}, 
we see that the difference field $(\set F(y),\sigma_y)$ is actually isomorphic to either the
difference field $(\set F(y),\varphi_{A_u})$ or the difference field $(\set F(y),\varphi_{A_q})$.
In other words, there exists an automorphism $\phi$ of $\set F(y)$ such that
\begin{equation}\label{EQ:iso}
\text{either}\quad \phi\circ \sigma_y = \varphi_{A_u} \circ \phi\quad\text{or}\quad
\phi\circ\sigma_y = \varphi_{A_q} \circ\phi.
\end{equation}

Now let $R$ be a ring extension of $\set F(y)$ and assume that the relation \eqref{EQ:iso} remains in $R$,
where $\sigma_y,\varphi_{A_u},\varphi_{A_q}$ are extended to be monomorphisms of~$R$, and $\phi$ is
extended to be an automorphism of $R$. 
An element $c\in R$ is called a {\em constant} if $\sigma_y(c) = c$.
All constants in $R$ form a subring of~$R$, denoted by $C_R$.

\begin{definition}\label{DEF:hypergeom}
An invertible element $T$ of $R$ is called a {\em $\sigma_y$-hypergeometric term} over $\set F(y)$ if 
$\sigma_y(T) = rT$ for some $r\in \set F(y)$. We call $r$ the {\em $\sigma_y$-quotient} of~$T$.
\end{definition}
Clearly, every nonzero rational function in $\set F(y)$ is $\sigma_y$-hypergeometric. 
A $\sigma_y$-hypergeometric term is also called a {\em hypergeometric term} if 
$\sigma_y$ is the usual shift operator, and a {\em $q$-hypergeometric term} if 
$\sigma_y$ is the $q$-shift operator. Typical examples are given by a Pochhammer symbol
$(a)_y = a(a+1)\cdots (a+y-1)\ (y>0)$ being a hypergeometric term in $y$ and its $q$-analogue 
$(q;q)_k = \prod_{i=1}^k (1-q^i)$ being a $q$-hypergeometric term in $q^k$.

Two $\sigma_y$-hypergeometric terms are called {\em similar} over $\set F(y)$ if one can be obtained from the other
by multiplying a rational function in $\set F(y)$.
A $\sigma_y$-hypergeometric term $T$ is said to be {\em $\sigma_y$-summable} if there exists
another $\sigma_y$-hypergeometric term $G$ such that $T = \Delta_y(G)$, where $\Delta_y$
denotes the difference of $\sigma_y$ and the identity map of~$R$.
It is readily seen that two $\sigma_y$-hypergeometric terms $T$ and $G$ satisfying $T = \Delta_y(G)$
are similar.  

Let $\sigma_1,\sigma_2$ be any two monomorphisms of~$R$ with $\phi\circ\sigma_1=\sigma_2\circ\phi$ 
for some automorphism $\phi$ of~$R$ whose restriction to $\set F(y)$ is an automorphism of~$\set F(y)$.
Then a term $T$ in $R$ is $\sigma_1$-hypergeometric if and only if $\phi(T)$ is 
$\sigma_2$-hypergeometric. Moreover, the problem of determining whether a 
$\sigma_1$-hypergeometric term $T$ is $\sigma_1$-summable is equivalent to that of 
determining whether the $\sigma_2$-hypergeometric term $\phi(T)$ is $\sigma_2$-summable.  
We thus conclude from \eqref{EQ:iso} that determining the $\sigma_y$-summability of a
$\sigma_y$-hypergeometric term amounts to determining the usual summability of a hypergeometric
term or the $q$-summability of a $q$-hypergeometric term.

For simplicity, in the rest of the paper, we assume throughout that, when restricted to $\set F(y)$,
the $\set F$-automorphism $\sigma_y$ is either the usual shift operator such that $\sigma_y(y) = y+1$ 
or the $q$-shift operator such that $\sigma_y(y) = qy$, where 
$q\in\set F$ is neither zero nor a root of unity. These two cases will be later referred to as the usual shift case 
and the $q$-shift case, respectively. We remark that in the case when $\sigma_y$ is the $q$-shift operator 
and $q$ is further assumed to be a root of unity, the $\sigma_y$-summability problem is closely related to 
the additive version of Hilbert's 
Theorem~90 (see~\cite[Theorem~6.3, Page~290]{Lang2002}), and will be left for future research.

\subsection{The canonical representation}
Let $T$ be a $\sigma_y$-hypergeometric term. 
A key idea on determining the $\sigma_y$-summability of a given 
$\sigma_y$-hypergeometric term $T$ is to write it into a multiplicative decomposition $T = fH$, where
$f\in \set F(y)$ and $H$ is a $\sigma_y$-hypergeometric term enjoying some nice properties 
(cf.\ \cite{AbPe2001a, AbPe2002b}). Then determining whether $T$ is $\sigma_y$-summable
amounts to finding a rational function $g\in \set F(y)$ such that 
\begin{equation}\label{EQ:summable}
fH = \Delta_y(gH),\ \text{or equivalently},\ f = K\sigma_y(g)-g\ \text{with}\ K = \frac{\sigma_y(H)}H.
\end{equation}
For a nonzero rational function $K$ in $\set F(y)$, we follow \cite{DHL2018} to
define a linear map $\Delta_K=K\sigma_y-{\bf 1}$
from $\set F(y)$ to itself which maps $r\in\set F(y)$ to $K\sigma_y(r)-r$.
Note that the image of $\Delta_K$, denoted by $\im(\Delta_K)$, is an 
$\set F$-linear subspace of $\set F(y)$.
It then follows from \eqref{EQ:summable} that $fH$ is $\sigma_y$-summable if and only if
$f\in \im(\Delta_K)$.
In this way, the main object has been reduced from $\sigma_y$-hypergeometric terms
to the well-studied class of rational functions. 

In the following, we adapt the notion of normal and special polynomials 
from symbolic integration \cite{Bron2005} into our setting, so as to obtain a canonical 
representation of a rational function.

\begin{definition}\label{DEF:types}
A polynomial $p\in \set F[y]$ is said to be {\em $\sigma_y$-normal} if $\gcd(p,\sigma_y^\ell(p)) = 1$ 
for any nonzero integer~$\ell$, and {\em $\sigma_y$-special} if $p\mid \sigma_y^\ell(p)$ for some 
nonzero integer $\ell$.
\end{definition}
Note that $\sigma_y$-normal polynomials are also called $\sigma_y$-free polynomials in the literature.
Recall that two polynomials $a,b$ in $\set F[y]$ are {\em associates} if $a = c\,b$ for $c\in \set F$.
Since $\sigma_y$ preserves the degree of the input polynomial, a polynomial $p\in \set F[y]$ 
is $\sigma_y$-special if and only if $p$ is an associate of $\sigma_y^\ell(p)$ for some nonzero integer~$\ell$, 
which happens if and only if $p$ is an associate of $\sigma_y^\ell(p)$ for some positive integer~$\ell$.
It is readily seen that $\sigma_y$-normal (resp.\ $\sigma_y$-special) polynomials remain to be 
$\sigma_y$-normal (resp.\ $\sigma_y$-special) under the application of the automorphism $\sigma_y$.
A polynomial is not necessarily $\sigma_y$-normal or $\sigma_y$-special, 
but an irreducible polynomial $p\in\set F[y]$ must be either $\sigma_y$-normal or 
$\sigma_y$-special, since $\gcd(p,\sigma_y^\ell(p))$ for any integer $\ell$ is a factor of $p$.
Evidently, all elements in $\set F$ are $\sigma_y$-special, and $p\in\set F[y]$ is 
both $\sigma_y$-normal and $\sigma_y$-special if and only if $p\in \set F\setminus\{0\}$. 

\begin{definition}\label{DEF:coprime}
Two polynomials $a,b\in\set F[y]$ are said to be {\em $\sigma_y$-coprime} if
$\gcd(a,\sigma_y^\ell(b)) = 1$ for any nonzero integer~$\ell$.
\end{definition}
Note that in the above definition, we did not require that two $\sigma_y$-coprime 
polynomials be coprime.
In analogy to \cite[Theorem~3.4.1]{Bron2005}, we describe the multiplicative properties of $\sigma_y$-special 
and $\sigma_y$-normal polynomials.
\begin{proposition}\leavevmode\null
\label{PROP:props}
\begin{itemize}
\item[(i)] Any finite product of $\sigma_y$-normal and pairwise $\sigma_y$-coprime 
polynomials in~$\set F[y]$ is $\sigma_y$-normal. Any factor of a $\sigma_y$-normal polynomial 
in $\set F[y]$ is $\sigma_y$-normal.

\item[(ii)] Any finite product of $\sigma_y$-special polynomials in $\set F[y]$ is $\sigma_y$-special.
Any factor of a nonzero $\sigma_y$-special polynomial in $\set F[y]$ is $\sigma_y$-special.
\end{itemize}
\end{proposition}
\begin{proof}
(i) Let $p_1,\dots,p_m\in\set F[y]$ be $\sigma_y$-normal and such that $\gcd(p_i,\sigma_y^k(p_j)) = 1$
for all $i,j,k\in \set Z$ with $1\leq i < j\leq m$ and $k\neq 0$, 
and let $p = \prod_{i=1}^mp_i$. Then for any nonzero
integer $\ell$, we have
\[
\gcd(p,\sigma_y^\ell(p)) 
= \gcd(p_1\cdots p_m, \sigma_y^\ell(p_1)\cdots\sigma_y^\ell(p_m))
\ \Big|\ \prod_{i=1}^m\gcd(p_i,\sigma_y^\ell(p_i)) = 1,
\]
where the last equality follows by the $\sigma_y$-normality of each $p_i$. 
Thus $\gcd(p,\sigma_y^\ell(p)) = 1$, that is, $p$ is $\sigma_y$-normal.

Let $p\in \set F[y]$ be $\sigma_y$-normal and write $p = ab$ where $a,b\in \set F[y]$. 
Since $p$ is $\sigma_y$-normal, we have 
\[
\gcd(p,\sigma_y^\ell(p))
= \gcd(ab,\sigma_y^\ell(a)\sigma_y^\ell(b)) = 1 
\quad\text{for all $\ell\in\set Z\setminus\{0\}$}.
\]
Thus $\gcd(a,\sigma_y^\ell(a)) = 1$ for all $\ell\in\set Z\setminus\{0\}$, which implies that $a$ is $\sigma_y$-normal.

\bigskip\noindent
(ii) Let $a,b\in\set F[y]$ be two $\sigma_y$-special polynomials. Then there exist 
positive integers $\ell_1,\ell_2$ and elements $c_1,c_2\in\set F$ such that 
$\sigma_y^{\ell_1}(a) = c_1a$ and $\sigma_y^{\ell_2}(b) = c_2b$. Thus
\[
\sigma_y^{\ell_1\ell_2}(ab) = \sigma_y^{\ell_1\ell_2}(a)\sigma_y^{\ell_1\ell_2}(b) 
= (c_1a)(c_2b) = c_1c_2 ab.
\]
So $ab\mid \sigma_y^{\ell_1\ell_2}(ab)$, that is, $ab$ is $\sigma_y$-special. 
The first assertion of part (ii) then follows by induction.

Let $p$ be a nonzero $\sigma_y$-special polynomial in $\set F[y]$. There is nothing to show if $p\in \set F$.
Assume that $p\notin \set F$ and let $a\in\set F[y]$ be an irreducible
factor of $p$. Then there exists a nonzero integer $\ell$ such that $p$ and $\sigma_y^\ell(p)$ are 
associates,
and thus $a\mid \sigma_y^\ell(p)$. It follows that there exists an irreducible factor $a_1\in\set F[y]$ of $p$
such that $a$ is an associate of $\sigma_y^\ell(a_1)$. Applying the same argument to $a_1$, we get an 
irreducible factor $a_2\in\set F[y]$ of $p$ such that $a_1$ is an associate of $\sigma_y^\ell(a_2)$. Thus,
$a$ is an associate of $\sigma_y^{2\ell}(a_2)$. 
Continuing in this pattern, we obtain a sequence of irreducible factors $\{a_1,a_2,\dots\}\subseteq \set F[y]$ 
of $p$ such that 
\[
a \ \text{is an associate of} \ \sigma_y^{k\ell}(a_k) \quad\text{for all}\ k= 1,2,\dots
\]
Since $p$ has only finitely many irreducible factors, there exist two integers $i,j$ with $j>i\geq 1$
such that $a_i = a_j$. Since $a$ is an associate of both $\sigma_y^{i\ell}(a_i)$ and $\sigma_y^{j\ell}(a_j)$, 
we see that
$a\mid \sigma_y^{(j-i)\ell}(a)$. Notice that $j-i>0$. So $(j-i)\ell\neq 0$. By definition,
$a$ is $\sigma_y$-special. Since $a$ is arbitrary, we know that every irreducible factor of 
$p$ is $\sigma_y$-special. Let now $b\in\set F[y]$ be any factor of $p$. If $b\in \set F$,
then $b$ is $\sigma_y$-special by definition. Otherwise, $b$ is a nonempty finite product 
of irreducible factors of $p$, so it is $\sigma_y$-special by the first assertion of part (ii).
\end{proof}

We can separate the $\sigma_y$-normal and $\sigma_y$-special components of a polynomial in $\set F[y]$.
\begin{definition}
Let $p \in \set F[y]$. We say that $p = p_s p_n$ is a {\em $\sigma_y$-splitting factorization} of~$p$ if
$p_s,p_n\in\set F[y]$, $p_s$ is $\sigma_y$-special, and every irreducible factor of $p_n$ is 
$\sigma_y$-normal.
\end{definition}

A consequence of Proposition~\ref{PROP:props} is that we always have $\gcd(p_s,p_n) = 1$
in a $\sigma_y$-splitting factorization $p=p_sp_n$ of $p\in\set F[y]$, and such a factorization is 
unique up to multiplication 
by units in $\set F$. Clearly, a full irreducible factorization of $p$ yields a $\sigma_y$-splitting 
factorization of $p$.
It is not obvious to us (but it would be interesting to see) whether such a $\sigma_y$-splitting
factorization can be computed by the gcd computation only, like in the differential case (cf.\ \cite[\S 3.5]{Bron2005}).

For a nonzero polynomial $p\in \set F[y]$, its degree in $y$ (or $y$-degree) is denoted by $\deg_y(p)$. 
We will follow the convention that $\deg_y(0) = -\infty$. 
We assume throughout that the numerator and denominator of a rational function
in $\set F(y)$ are always coprime.
A rational function in $\set F(y)$ is said to be {\em proper} if the $y$-degree of its numerator 
is less than that of its denominator.

We can now define a canonical representation of rational functions in $\set F(y)$.
Let $f$ be a rational function in~$\set F(y)$ with denominator $d$ and let $d = d_sd_n$ 
be a $\sigma_y$-splitting factorization of~$d$. Then there are unique $p,a,b\in\set F[y]$
such that $\deg_y(a)<\deg_y(d_s), \deg_y(b)<\deg_y(d_n)$, and
\[
f = p + \frac{a}{d_s}+\frac{b}{d_n}.
\]
We call this decomposition, which is unique, the {\em $\sigma_y$-canonical representation} of $f$,
and the components $p,a/d_s,b/d_n$ the {\em polynomial part}, the {\em special part}, the 
{\em normal part} of $f$, respectively.
Note that the special and normal parts of a rational function in $\set F(y)$ are always proper.

In the usual shift and the $q$-shift cases, 
we are able to characterize all possible $\sigma_y$-special irreducible polynomials.
To this end, we need the following simple lemma.
\begin{lemma}\label{LEM:invariant}
Let $f$ be a rational function in $\set F(y)$.
\begin{itemize}
\item[(i)] If $f(y+\ell) = f(y)$ for some nonzero integer $\ell$, then $f\in\set F$.
\item[(ii)] If $f(q^\ell y) = cf(y)$ for some nonzero integer $\ell$ and $c\in\set F$,
then $f/y^k\in\set F$ for some $k\in \set Z$. In particular, $c = q^{\ell k}$ if $f\neq 0$.
\end{itemize}
\end{lemma}
\begin{proof}
(i) This is exactly \cite[Lemma~2]{AbPe2002a}.

\smallskip\noindent
(ii) It is trivial when $f = 0$. Assume that $f$ is nonzero and $f(q^\ell y) = cf(y)$ for some $\ell\in\set Z\setminus\{0\}$ and $c\in\set F$. Write $f = a/d$ with $a,d\in \set F[y]\setminus\{0\}$ and $\gcd(a,d) = 1$.
It then follows from $f(q^\ell y) = cf(y)$ that
\[
a(q^\ell y) = c_1 a(y) \quad\text{and}\quad d(q^\ell y) = c_2 d(y),
\]
where $c_1,c_2\in \set F$ with $c_1/c_2 = c$. Let $a = \sum_{i=0}^{k_1}a_iy^i$ with $k_1\in\set N$,
$a_i\in \set F$ and $a_{k_1}\neq 0$. By comparing the coefficients of both sides of $a(q^\ell y) = c_1 a(y)$,
we conclude that
\[
a = a_{k_1}y^{k_1}\quad\text{and}\quad c_1 = q^{\ell k_1}.
\]
Similarly, we have $d = d_{k_2} y^{k_2}$ and $c_2 = q^{\ell k_2}$ for some $k_2\in \set N$ and 
$d_{k_2}\in \set F\setminus\{0\}$. Letting $k = k_1-k_2\in\set Z$, we obtain that $f/y^k\in \set F$ and $c = q^{\ell k}$.
\end{proof}

The following lemma gives the desired characterization.
\begin{lemma}\label{LEM:special}
Let $p$ be a polynomial in $\set F[y]$. 
\begin{itemize}
\item[(i)] In the usual shift case, $p$ is $\sigma_y$-special if and only if $p\in \set F$.
\item[(ii)] In the $q$-shift case, $p$ is $\sigma_y$-special 
if and only if $p$ is an associate of $y^k$ for some $k\in \set N$.
\end{itemize}
\end{lemma}
\begin{proof}
(i) The sufficiency is evident by definition. For the necessity, assume that $p$ is $\sigma_y$-special.
Then there exist $\ell\in\set Z\setminus\{0\}$ and $c\in \set F$ such that $\sigma_y^\ell(p)=c\,p$. 
Since $\sigma_y$ is the usual shift operator, $\sigma_y^\ell(p)$ and $p$ 
have the same leading coefficient with respect to~$y$. Thus $c=1$ and $\sigma_y^\ell(p) = p$. 
It follows from part (i) of Lemma~\ref{LEM:invariant} that $p\in \set F$.

\smallskip\noindent
(ii) Assume that $p = cy^k$ for some $c\in \set F$ and $k\in \set N$. 
Since $\sigma_y$ is the $q$-shift operator, we have $\sigma_y(p) = c\,q^k y^k$. 
So $p\mid \sigma_y(p)$, and thus $p$ is $\sigma_y$-special by definition.
Conversely, assume that $p$ is $\sigma_y$-special. Then there exist $\ell\in\set Z\setminus\{0\}$ and 
$c\in \set F$ such that $\sigma_y^\ell(p)=c\,p$. The assertion follows 
from part (ii) of Lemma~\ref{LEM:invariant}.
\end{proof}

\subsection{Kernels, shells and $\sigma_y$-factorizations}\label{SEC:ksf}
Recall that a polynomial $p$ in $\set F[y]$ is said to be {\em monic} if its leading coefficient with respect to~$y$
is one, and {\em $q$-monic} if $p(0) = 1$ (cf.\ \cite{PaRi1997}). Unifying the usual shift and the $q$-shift cases,
we say that a polynomial $p$ in $\set F[y]$ is {\em $\sigma_y$-monic} if it is monic in the former case 
or $q$-monic in the latter one. A rational function in~$\set F(y)$ is {\em $\sigma_y$-monic} if both 
its numerator and denominator are $\sigma_y$-monic.
By a factor of a rational function in~$\set F(y)$, we mean a factor of either its numerator or its denominator. 
Let $f\in \set F(y)$ be $\sigma_y$-monic. Lemma~\ref{LEM:special} then tells us that all irreducible factors of~$f$ 
are $\sigma_y$-normal. Moreover, $\sigma_y^\ell(f)$ for all $\ell\in \set Z$ is again $\sigma_y$-monic. 

Based on~\cite{AbPe2002b,CHM2005}, a nonzero rational function in $\set F(y)$ with numerator
$u$ and denominator~$v$ is said to be {\em $\sigma_y$-reduced} if $u$ and $v$ are $\sigma_y$-coprime.
For a nonzero rational function $f\in \set F(y)$, there exist
two nonzero rational functions $K,S$ in $\set F(y)$ with $K$ being $\sigma_y$-reduced such that
\[
f  = K\frac{\sigma_y(S)}S.
\]
Such a pair $(K,S)$ will be called a {\em rational normal form} (or an {\em RNF} for short) of $f$.
Moreover, we call $K$ a {\em kernel} and $S$ a corresponding {\em shell} of $f$. 
These quantities can be constructed by gcd-calculations (cf.~\cite{AbPe2002b,CHM2005}).

It is known that a rational function in $\set F(y)$ is a Laurent polynomial in~$y$ if its denominator
is a power of $y$. All Laurent polynomials in $\set F(y)$ form a subring, which is denoted by
$\set F[y,y^{-1}]$. Let $f$ be a nonzero Laurent polynomial in~$\set F[y,y^{-1}]$.
Then it can be written in the form $f = \sum_{i=m}^nc_iy^i$, where $m,n\in\set Z$ 
with $m\leq n$ and $c_m,c_{m+1},\dots,c_n\in \set F$ with $c_mc_n\neq 0$.
We call $n$ the {\em head degree} of $f$ and $m$ the {\em tail degree} of $f$, which are 
denoted by $\hdeg(f)$ and $\tdeg(f)$, respectively. By convention, we have
$\hdeg(0) = -\infty$ and $\tdeg(0) = +\infty$.

We will also consider the ring of Laurent polynomials in $\sigma_y$ over $\set Z$, denoted by
$\set Z[\sigma_y,\sigma_y^{-1}]$. Let $p$ be a nonzero polynomial in~$\set F[y]$ and let
$\alpha = \sum_{i=m}^nk_i\sigma_y^i\in \set Z[\sigma_y,\sigma_y^{-1}]$. We define
\[
p^\alpha = \prod_{i=m}^n\sigma_y^i(p)^{k_i}.
\]
Clearly, $p^\alpha$ is a polynomial if and only if $\alpha$ belongs to $\set N[\sigma_y,\sigma_y^{-1}]$.

According to \cite[Definition~11]{Karr1981} and \cite[Definition~1]{AbPe2002b}, 
two polynomials $a,b\in\set F[y]$ are {\em $\sigma_y$-equivalent}
if $a$ is an associate of $\sigma_y^\ell(b)$ for some integer $\ell$.
Evidently, this gives an equivalence relation,
and the $\sigma_y$-equivalence of two polynomials can be easily recognized 
by comparing coefficients. Let $f$ be a rational function in~$\set F(y)$.
By computing $\sigma_y$-splitting factorizations of its numerator and denominator,
and grouping together all $\sigma_y$-normal irreducible factors that are $\sigma_y$-equivalent,
we can decompose $f$ as
\begin{equation}\label{EQ:sigmafac}
f =f_s\, p_1^{\alpha_1}\dots\, p_m^{\alpha_m},
\end{equation}
where $f_s\in \set F(y)$ whose numerator and denominator are both $\sigma_y$-special,
$m\in \set N$, each $\alpha_i\in \set Z[\sigma_y,\sigma_y^{-1}]\setminus\{0\}$, each
$p_i\in \set F[y]$ is $\sigma_y$-monic and irreducible,
and the $p_i$ are pairwise $\sigma_y$-inequivalent.
We call \eqref{EQ:sigmafac} a {\em $\sigma_y$-factoriztion} of~$f$.
Note that such a factorization is not unique since there are many possibilities to 
express each component $p_i^{\alpha_i}$. Nevertheless,
for each fixed $p_i$, the corresponding exponent $\alpha_i$ in~\eqref{EQ:sigmafac} is unique 
as $p_i$ is $\sigma_y$-normal and $\set F[y]$ is a unique factorization domain, and 
we will then call $\alpha_i$ the {\em $\sigma_y$-exponent} of $p_i$ in~$p$.

The following describes a useful property of $\sigma_y$-reduced rational functions, 
which is equivalent to \cite[Lemma~2.2]{CHKL2015} in the usual shift case and
\cite[Proposition~3.2]{DHL2018} in the $q$-shift case.
\begin{proposition}\label{PROP:redprop}
Let $r$ be a $\sigma_y$-reduced rational function in~$\set F(y)$, and assume that
$r = \sigma_y(f)/f$ for some $f\in \set F(y)\setminus\{0\}$. Then both the numerator
and denominator of $f$ are $\sigma_y$-special. Moreover, $r$ is equal to one in 
the usual shift case, or it is a power of $q$ in the $q$-shift case.
\end{proposition}
\begin{proof}
Assume that $f$ admits a $\sigma_y$-factorization of the form~\eqref{EQ:sigmafac}
and suppose that $m>0$. Since $r = \sigma_y(f)/f$, we get
\[
r = \frac{\sigma_y(f_s)}{f_s}\prod_{i=1}^mp_i^{\beta_i},
\]
where $\beta_i = \sigma_y\alpha_i - \alpha_i\neq 0$ for all $i=1,\dots,m$.
Notice that the total sum of all coefficients of each $\beta_i$ with respect to $\sigma_y$
is zero. So each $\beta_i$ has both positive and negative coefficients, which contradicts
with the assumption that $r$ is $\sigma_y$-reduced. Therefore, $m=0$ and then $f = f_s$,
implying that both the numerator and denominator of $f$ are $\sigma_y$-special.
The second assertion immediately follows by Lemma~\ref{LEM:special}.
\end{proof}

The above property enables us to derive the following equivalence characterization of 
{\em rational $\sigma_y$-hypergeometric terms}, that is, terms
of the form $cf$, where $c\in C_R\setminus\{0\}$ and $f\in\set F(y)$.
\begin{corollary}\label{COR:rathyper}
Let $T\in R$ be a $\sigma_y$-hypergeometric term whose $\sigma_y$-quotient has a kernel $K$.
\begin{itemize}
\item[(i)] In the usual shift case, $T$ is rational if and only if $K=1$.
\item[(ii)] In the $q$-shift case, $T$ is rational if and only if $K$ is a power of $q$.
\end{itemize}
\end{corollary}
\begin{proof}
Let $S$ be a shell of $\sigma_y(T)/T$ so that $\sigma_y(T)/T = K\sigma_y(S)/S$.
Assume that $T=cf$ for some $c\in C_R\setminus\{0\}$ and $f\in \set F(y)$. Then 
$\sigma_y(T)/T = \sigma_y(f)/f = K\sigma_y(S)/S$.
Thus $K = \sigma_y(r)/r$, where $r = f/S\in \set F(y)$.
Since $K$ is $\sigma_y$-reduced, we see from Proposition~\ref{PROP:redprop}
that $K$ is equal to one in the usual shift case or it is a power of $q$
in the $q$-shift case.
The necessities of both parts (i) and (ii) thus follow.

Conversely, assume that $K = 1$ in the usual shift case or $K = q^k$ for 
some $k\in \set Z$ in the $q$-shift case. By taking $r = 1$ in the former case
or $r = y^k$ in the latter one, we obtain that $K = \sigma_y(r)/r$ in either case. 
Notice that $\sigma_y(T)/T = K\sigma_y(S)/S$.
Thus $T/(rS)$ is a nonzero constant, say $c$, of the ring $R$. 
It follows that $T = c\, rS$. 
\end{proof}
As mentioned in the paragraph right after Corollary~3.2 in \cite{DHL2018},
the ring $R$ can be chosen using Picard-Vessiot extensions (cf.~\cite{BLW2005,HaSi2008})
so that $C_R$ coincides with the field $\set F$ if $\set F$ is further assumed to be algebraically closed. 

\section{A unified reduction}\label{SEC:red}
Let $T$ be a $\sigma_y$-hypergeometric term whose $\sigma_y$-quotient has
an RNF $(K,S)$. According to \cite{AbPe2002b,CHM2005},
$T$ admits a {\em multiplicative decomposition} $SH$, where $H$ is another 
$\sigma_y$-hypergeometric term  with $\sigma_y$-quotient $K$.
We are going to reduce the shell $S$ modulo $\im(\Delta_K)$ to a rational function $r\in \set F(y)$,
which is minimal in some sense. This reduction gives rise to an additive decomposition 
$T = \Delta_y(gH)+ rH$ for some $g\in \set F(y)$. 
The minimality of $r$ will then establish a $\sigma_y$-summability criterion which says that 
$T$ is $\sigma_y$-summable if and only if $r = 0$ (see Theorem~\ref{THM:shellred}).

Using an arbitrary RNF, the task described above can be accomplished by 
a shell reduction enhanced with a polynomial reduction in both the usual shift case~\cite{CHKL2015} 
and the $q$-shift case~\cite{DHL2018}.
The difference is that, in the $q$-shift case, we need to reduce Laurent polynomials instead of 
polynomials, which complicates the steps for the polynomial reduction.
In order to force the $q$-shift case to be in line with the usual shift case as much as possible, 
inspired by \cite{PaSt1995}, we will introduce the notion of $\sigma_y$-standard rational functions 
(see Definition~\ref{DEF:standard}) and further show that every nonzero rational function in $\set F(y)$
has a $\sigma_y$-standard kernel (see Proposition~\ref{PROP:kernel}).
For a $\sigma_y$-hypergeometric term whose $\sigma_y$-quotient has a $\sigma_y$-standard kernel $K$
and a corresponding shell~$S$,
the main steps of our reduction algorithm are as follows: first write the shell $S$ in its canonical 
representation and then perform successive reductions each of which brings the individual part to a 
\lq\lq simple\rq\rq\ one, until the remaining rational function is \lq\lq minimal\rq\rq. 

\subsection{Normal reduction}\label{SUBSEC:normalred}
In this subsection, we aim at reducing the normal part of a rational function to a \lq\lq simple\rq\rq\ one.
The main strategy differs slightly from the method presented in \cite[\S 4]{DHL2018} in that it employs the
so-called strong $\sigma_y$-factorization, rather than an arbitrary one, so as to make the process more
concise and more trackable.
The following definition is useful in justifying the simplicity.
\begin{definition}
Let $K\in \set F(y)$ with numerator $u$ and denominator $v$. 
A nonzero polynomial $p\in \set F[y]$ is said to be {\em strongly coprime} with $K$ if 
$\gcd(u,\sigma_y^\ell(p)) = \gcd(v,\sigma_y^{-\ell}(p)) = 1$ for all $\ell\in \set N$.
\end{definition}
The next lemma is  used to verify the minimality of our reduction algorithm, which
applies to both the usual shift and the $q$-shift cases, and thus extends \cite[Lemma~4.1]{DHL2018}.
\begin{lemma}\label{LEM:minimal}
Let $K\in \set F(y)$ with denominator~$v$, 
and let $h\in \set F(y)$ be a rational function whose denominator $d$ is $\sigma_y$-normal 
and strongly coprime with~$K$. Assume that there are $\tilde h\in \set F(y)$ 
and $p\in \set F[y]$ such that
\begin{equation}\label{EQ:minimal}
h-\tilde h + \frac{p}v\in \im(\Delta_K).
\end{equation}
Then the degree of $d$ is no more than that of the denominator of $\tilde h$. 
\end{lemma}
\begin{proof}
Write $K = u/v$, where $u\in \set F[y]$ with $\gcd(u,v) = 1$.
By \eqref{EQ:minimal}, there exists $g\in \set F(y)$ such that
\[
h-\tilde h + \frac{p}v = K\sigma_y(g)-g.
\]
Multiplying both sides by $v$ yields
\begin{equation}\label{EQ:vminimal}
v(h - \tilde h) -(u\sigma_y(g) - v g) = -p\in \set F[y].
\end{equation}
There is nothing to show if $d\in \set F$. Assume that $d\notin \set F$ and let
$a\in \set F[y]$ be an irreducible factor of~$d$ with multiplicity $k$. 
Since $d$ is $\sigma_y$-normal, $a$ is $\sigma_y$-normal as well by Proposition~\ref{PROP:props} (i).
Notice that all irreducible factors of $d$ are mutually $\sigma_y$-inequivalent.
So it suffices to show that there exists an integer $\ell$ such that $\sigma_y^\ell(a)^k$ 
divides the denominator $\tilde d$ of $\tilde h$.
Suppose that $a^k$ does not divide $\tilde d$,
otherwise we have done. Notice that $a\nmid v$ as $d$ 
is strongly coprime with $K$. It thus follows from \eqref{EQ:vminimal} that $a^k$
divides either $e$ or $\sigma_y(e)$, where $e$ is the denominator of~$g$.

If $a^k\mid e$, then there exists an integer $\ell\geq 1$ such that $\sigma_y^{\ell-1}(a)^k\mid e$
but $\sigma_y^\ell(a)^k\nmid e$ since $a$ is $\sigma_y$-normal. 
Thus we have $\sigma_y^\ell(a)^k\mid \sigma_y(e)$.
Since $d$ is $\sigma_y$-normal and strongly coprime with $K$, neither $d$ nor $u$ is 
divisible by $\sigma_y^\ell(a)$. Therefore, we conclude from \eqref{EQ:vminimal} that 
$\sigma_y^\ell(a)^k$ divides $\tilde d$.

If $a^k\mid \sigma_y(e)$, then there exists an integer $\ell\leq -1$ such that 
$\sigma_y^\ell(a)^k\mid e$ but $\sigma_y^{\ell-1}(a)^k\nmid e$ since $a$ is $\sigma_y$-normal. Thus we have
$\sigma_y^\ell(a)^k\nmid \sigma_y(e)$. Similarly, neither $d$ nor $v$ is divisible by
$\sigma_y^\ell(a)$, because $d$ is $\sigma_y$-normal and strongly coprime with $K$.
Therefore, $\sigma_y^\ell(a)^k$ divides $\tilde d$ by~\eqref{EQ:vminimal}.
\end{proof}
Let $p\in \set F[y]$ be $\sigma_y$-normal and irreducible, and 
$\alpha \in\set N[\sigma_y,\sigma_y^{-1}]\setminus\{0\}$. Let $K\in\set F(y)$ be $\sigma_y$-reduced
with numerator $u$ and denominator $v$. 
Let $\lambda,\mu\in\set N[\sigma_y,\sigma_y^{-1}]$ be the $\sigma_y$-exponents of~$p$ 
in $u$ and $v$, respectively. Since $K$ is $\sigma_y$-reduced, at least one of
$\lambda$ and $\mu$ is zero, that is, we have $\lambda\mu = 0$. Define
\begin{equation}\label{EQ:stronglycoprime}
\ell =\begin{cases}
\tdeg(\mu)-1\ & \text{if}\ \lambda = 0,\\[1ex]
\hdeg(\lambda) + 1 & \text{otherwise}.
\end{cases}
\end{equation}
Then we rewrite
\[
p^\alpha = (p^{\sigma_y^\ell})^{\sigma_y^{-\ell} \alpha},
\]
where $p^{\sigma_y^\ell}$ is strongly coprime with $K$ and $\sigma_y^{-\ell}\alpha\in\set N[\sigma_y,\sigma_y^{-1}]\setminus\{0\}$. 
Note that $p^{\sigma_y^\ell}$ is again $\sigma_y$-normal and irreducible.
This implies that every $\sigma_y$-normal and irreducible polynomial in~$\set F[y]$ can be 
transformed to one which is 
$\sigma_y$-equivalent to the original polynomial and is strongly coprime with~$K$.
Such a transformation enables
us to obtain a more structured $\sigma_y$-factorization of a given polynomial in $\set F[y]$.

Consider now a polynomial $p\in \set F[y]$ with a $\sigma_y$-factorization 
$p =p_s \prod_{i=1}^mp_i^{\alpha_i}$ and let $K\in\set F(y)$ be a $\sigma_y$-reduced rational function.
For each factor $p_i$, we transform it to one which is strongly coprime with $K$ using 
the procedure described in the preceding paragraph. By relabeling all the resulting factors, 
we finally arrive at the following decomposition (with a slight abuse of notation)
\begin{equation}\label{EQ:sigmasfac}
p =p_s\, p_1^{\alpha_1}\dots\, p_m^{\alpha_m},
\end{equation}
where $m\in \set N$ and 
\begin{itemize}
\item $p_s\in\set F[y]$ is $\sigma_y$-special;
\item each $p_i\in \set F[y]$ is $\sigma_y$-monic, irreducible and strongly coprime with $K$;
\item the $p_i$ are pairwise $\sigma_y$-inequivalent;
\item each $\alpha_i$ is in $\set N[\sigma_y,\sigma_y^{-1}]\setminus\{0\}$.
\end{itemize}
We will call \eqref{EQ:sigmasfac} a {\em strong $\sigma_y$-factorization} of~$p$ with respect to~$K$.

Before turning to the general case, we first perform the normal reduction ``locally".
\begin{lemma}\label{LEM:normalmonomial}
Let $K\in \set F(y)$ with denominator~$v$, 
and let $f\in\set F(y)$ be a nonzero proper rational function with denominator $d^{k\sigma_y^\ell}$, 
where $k\in\set Z^+$, $\ell\in \set Z$ and $d\in \set F[y]$ is strongly coprime with~$K$. 
Then there exist $g\in \set F(y)$ and $a,b\in \set F[y]$ with $\deg_y(a)<k\deg_y(d)$ such that
\begin{equation}\label{EQ:normalmonomial}
f =\Delta_K(g) + \frac{a}{d^k} + \frac{b}v.
\end{equation}
\end{lemma}
\begin{proof}
Write $K = u/v$ and $f = c/d^{k\sigma_y^\ell}$, where $u,c\in \set F[y]$ with $\gcd(u,v) = 1$ and $\gcd(c,d^{k\sigma_y^\ell}) = 1$.
If $\ell=0$ then letting $g=0$, $a = c$ and $b=0$ immediately yields the assertion.
Now assume that $\ell$ is nonzero. So it is either positive or negative.

If $\ell>0$, then $\gcd(u,d^{k\sigma_y^\ell}) = 1$ since $d$ is strongly coprime with~$K$.
Using the extended Euclidean algorithm, we can find $s,t\in \set F[y]$ 
with $\deg_y(s)<k\deg_y(d)$ such that
\[
vc = s u + t d^{k\sigma_y^\ell}.
\]
Multiplying both sides by $1/(vd^{k\sigma_y^\ell})$ gives
\[
f = \frac{vc}{vd^{k\sigma_y^\ell}} 
= K\frac{s}{d^{k\sigma_y^\ell}}+ \frac{t}{v}.
\]
Adding and subtracting $\sigma_y^{-1}(s)/d^{k\sigma_y^{\ell-1}}$ to the right-hand side, 
we get
\begin{equation}\label{EQ:positive}
f = K\sigma_y\Big(\frac{\sigma_y^{-1}(s)}{d^{k\sigma_y^{\ell-1}}}\Big)
- \frac{\sigma_y^{-1}(s)}{d^{k\sigma_y^{\ell-1}}}
+\frac{\sigma_y^{-1}(s)}{d^{k\sigma_y^{\ell-1}}} + \frac{t}{v}
=\Delta_K(g_0)+\tilde  f+\frac{t}v,
\end{equation}
where $g_0 = \sigma_y^{-1}(s)/d^{k\sigma_y^{\ell-1}}$ and 
$\tilde f = \sigma_y^{-1}(s)/d^{k\sigma_y^{\ell-1}}$.
If $\tilde f = 0$, the proof is then concluded by letting $g = g_0=0$, $a=0$ and $b = t$.
Otherwise, $\tilde f$ is a nonzero proper rational function with denominator $d^{k\sigma_y^{\ell-1}}$,
so by induction on~$\ell$, we can find $\tilde g\in \set F(y)$ and $a,\tilde b\in \set F[y]$
with $\deg_y(a) <k\deg_y(d)$ such that
\[
\tilde f = \Delta_K(\tilde g) + \frac{a}{d^k} + \frac{\tilde b}v,
\]
which, along with \eqref{EQ:positive}, establishes \eqref{EQ:normalmonomial}
with $g=g_0+\tilde g$ and $b = t+\tilde b$.

If $\ell<0$, then $\gcd(v,d^{k\sigma_y^{\ell+1}}) = 1$ since $d$ is strongly coprime with~$K$.
Again, we can employ the extended Euclidean algorithm to find $s,t\in \set F[y]$ with
$\deg_y(s)<k\deg_y(d)$ such that
\[
u\sigma_y(c) = s v + t d^{k\sigma_y^{\ell+1}}. 
\]
Multiplying both sides by $1/(vd^{k\sigma_y^{\ell+1}})$ gives
\[
K\sigma_y(f) =\frac{u\sigma_y(c)}{vd^{k\sigma_y^{\ell+1}}}
= \frac{s}{d^{k\sigma_y^{\ell+1}}} + \frac{t}v.
\]
Adding and subtracting $K\sigma_y(f)$ to $f$, we get
\begin{equation}\label{EQ:negative}
f = K\sigma_y(-f) - (-f) + K\sigma_y(f) 
= K\sigma_y(-f) - (-f) + \frac{s}{d^{k\sigma_y^{\ell+1}}} + \frac{t}v 
=\Delta_K(g_0) + \tilde f + \frac{t}v,
\end{equation}
where $g_0 = -f$ and $\tilde f = s/d^{k\sigma_y^{\ell+1}}$. 
If $\tilde f = 0$, then the assertion follows by letting $g=g_0 = -f$, $a = 0$ and $b = t$.
Otherwise, $\tilde f$ is a nonzero proper rational function with denominator $d^{k\sigma_y^{\ell+1}}$,
so by induction on $\ell$, we can find $\tilde g\in \set F(y)$ and $a,\tilde b\in \set F[y]$
with $\deg_y(a) <k\deg_y(d)$ such that
\[
\tilde f = \Delta_K(\tilde g) + \frac{a}{d^k} + \frac{\tilde b}v,
\]
which, along with \eqref{EQ:negative}, establishes \eqref{EQ:normalmonomial}
with $g=g_0+\tilde g$ and $b = t+\tilde b$.
\end{proof}
\begin{lemma}\label{LEM:normallocal}
Let $K\in \set F(y)$ with denominator~$v$, 
and let $f\in\set F(y)$ be a nonzero proper rational function with denominator $d^\alpha$, 
where $d\in \set F[y]$ is $\sigma_y$-normal and strongly coprime with $K$ 
and $\alpha\in \set N[\sigma_y,\sigma_y^{-1}]\setminus\{0\}$.
Then there exist $g\in \set F(y)$, $a,b\in \set F[y]$ and $k\in \set N$
with $\deg_y(a)<k\deg_y(d)$ such that
\begin{equation}\label{EQ:normallocal}
f =\Delta_K(g) + \frac{a}{d^k} + \frac{b}v.
\end{equation}
\end{lemma}
\begin{proof}
Assume that $\alpha = \sum_{i=m}^n k_i\sigma_y^i$, where $m,n\in \set Z$, $m\leq n$, 
$k_i\in \set N$ and $k_mk_n\neq 0$. Since $d$ is $\sigma_y$-normal, the polynomials
$d^{\sigma_y^m}, d^{\sigma_y^{m+1}}, \dots, d^{\sigma_y^n}$ are pairwise coprime.
Then $f$ admits a partial fraction decomposition $f = \sum_{i=m}^nf_i$, where each $f_i$
is either zero or a nonzero proper rational function in $\set F(y)$ with denominator
$d^{k_i\sigma_y^i}$. Applying Lemma~\ref{LEM:normalmonomial} to each nonzero 
$f_i$ yields
\[
f_i = \Delta_K(g_i) + \frac{a_i}{d^{k_i}} + \frac{b_i}v,
\]
where $g_i\in \set F(y)$, $a_i,b_i\in \set F[y]$ and $\deg_y(a_i)<k_i \deg_y(d)$.
Summing all these equations up, we thus obtain \eqref{EQ:normallocal} for some
$g\in\set F(y)$, $a,b\in \set F[y]$ and $k\in \set N$ satisfying $\deg_y(a)<k\deg_y(d)$.
\end{proof}
The main result of this subsection is given below.
\begin{theorem}\label{THM:normal}
Let $K\in \set F(y)$ be $\sigma_y$-reduced with denominator~$v$, 
and let $f\in\set F(y)$ be a proper rational function whose denominator does not have
any nontrivial $\sigma_y$-special factors.
Then there exist $g,h\in \set F(y)$ and $b\in \set F[y]$ such that 
\begin{equation}\label{EQ:normaldecomp}
f= \Delta_K(g) + h + \frac{b}v,
\end{equation}
and $h$ is proper whose denominator is $\sigma_y$-normal and strongly coprime with $K$.
Moreover, the denominator of $h$ has minimal $y$-degree in the sense that if
there exists another triple $(\tilde g, \tilde h, \tilde b)$ with $\tilde g,\tilde h \in \set F(y)$ 
and $\tilde b\in \set F[y]$ such that
\begin{equation}\label{EQ:normaldecomp2}
f = \Delta_K(\tilde g) + \tilde h + \frac{\tilde b}v,
\end{equation}
then the $y$-degree of the denominator of $h$ is no more than that of $\tilde h$.
In particular, $h$ is equal to zero if $f\in \im(\Delta_K)$.
\end{theorem}
\begin{proof}
If $f =0$, then the assertion is evident by letting $g = h = b = 0$. Assume that $f$ is nonzero and
write $f = a/d$ with $a,d\in \set F[y]$, $\gcd(a,d) = 1$ and $\deg_y(a)<\deg_y(d)$. 
Let $d=d_s\prod_{i=1}^md_i^{\alpha_i}$ be a strong $\sigma_y$-factorization of $d$ with respect to~$K$,
where $d_s\in \set F$ by assumption. Then $f$ has a partial fraction decomposition 
\begin{equation}\label{EQ:pfd}
f = \sum_{i=1}^mf_i,
\end{equation}
where $f_i\in \set F(y)$ is a nonzero proper rational function with denominator $d_i^{\alpha_i}$ 
for $i = 1,\dots,m$. For all $i = 1,\dots,m$, we can apply Lemma~\ref{LEM:normallocal} to $f_i$ 
to find $g_i\in \set F(y)$, $a_i,b_i\in \set F[y]$ and $k_i\in\set N$ with 
$\deg_y(a_i)<k_i\deg_y(d_i)$ such that
\begin{equation}\label{EQ:fi}
f_i = \Delta_K(g_i)+ \frac{a_i}{d_i^{k_i}} + \frac{b_i}v.
\end{equation}
Then \eqref{EQ:normaldecomp} follows by 
letting $g=\sum_{i=1}^mg_i$, $h = \sum_{i=1}^ma_i/d_i^{k_i}$ and $b = \sum_{i=1}^mb_i$.
Note that the irreducible polynomials $d_1,\dots,d_m$ are $\sigma_y$-normal and 
mutually $\sigma_y$-inequivalent. So they are two by two $\sigma_y$-coprime, and 
thus the denominator of $h$ is $\sigma_y$-normal by Proposition~\ref{PROP:props}~(i).
Because $d_1,\dots,d_m$ are all strongly coprime with $K$, so is the denominator of $h$.
Moreover, $h$ is proper since all the $a_i/d_i^{k_i}$ are proper. 

It remains to show that the $y$-degree of the denominator of $h$ is minimal.
Assume that there exist $\tilde g,\tilde h\in \set F(y)$ and $\tilde b\in \set F[y]$
such that \eqref{EQ:normaldecomp2} holds. Then by \eqref{EQ:normaldecomp},
we have
\[
h - \tilde h + \frac{b-\tilde b}v\in \im(\Delta_K).
\]
It follows from Lemma~\ref{LEM:minimal} that the $y$-degree of the denominator of $h$
is no more than that of~$\tilde h$. Now assume that $f\in \im(\Delta_K)$. 
Then \eqref{EQ:normaldecomp2} holds with $\tilde h = \tilde b = 0$ and thus $h\in \set F[y]$
by the minimality. Since $h$ is proper, it must be zero.
\end{proof}
The proof of Theorem~\ref{THM:normal} induces an algorithm
as follows.

\smallskip\noindent
{\bf NormalReduction.} Given a $\sigma_y$-reduced rational function
$K\in \set F(y)$ with denominator $v$, and a proper rational function $f\in\set F(y)$ 
whose denominator does not have any nontrivial $\sigma_y$-special factors, 
compute two rational functions $g,h\in \set F(y)$ and a polynomial $b\in \set F[y]$ 
such that \eqref{EQ:normaldecomp} holds and $h$ is proper whose denominator is 
$\sigma_y$-normal and strongly coprime with~$K$.
\begin{enumerate}
\item If $f = 0$, then set $g=0, h=0, b = 0$, and return.
\item Compute a strong $\sigma_y$-factorization $d=d_s\prod_{i=1}^md_i^{\alpha_i}$ of the
denominator $d$ of $f$ with respect to~$K$.
\item Compute the partial fraction decomposition \eqref{EQ:pfd} of $f$
with respect to $d=d_s\prod_{i=1}^md_i^{\alpha_i}$.
\item For $i= 1,\dots, m$, apply Lemma~\ref{LEM:normallocal} to $f_i$ to find
$g_i\in \set F(y)$, $a_i,b_i\in \set F[y]$ and $k_i\in\set N$ with 
$\deg_y(a_i)<k_i\deg_y(d_i)$ such that \eqref{EQ:fi} holds.
\item Set 
$
g=\sum_{i=1}^mg_i,  h = \sum_{i=1}^m a_i/d_i^{k_i},
b = \sum_{i=1}^mb_i,
$
and return.
\end{enumerate}
\begin{example}\label{EX:normalred}
Assume that $\sigma_y$ is the $q$-shift operator. Let $K = -qy+1$, which is $\sigma_y$-reduced, 
and let
\[
f=-\frac{q(q-1)y}{(qy - 1)(q^2y - 1)},
\]
whose denominator has no nontrivial $\sigma_y$-special factors, and admits a 
strong $\sigma_y$-factorization $(qy - 1)(q^2y - 1)=d^\alpha$ with $d = -q^2y+1$ and 
$\alpha = \sigma_y^{-1}+1$. Applying the normal reduction to $f$ with respect to $K$ 
yields 
\[
f = \Delta_K\Big(\frac1{qy-1}\Big) + \frac{-1/q}{d}+\frac{1/q}{v},
\]
where $v = 1$. Since the second summand is nonzero, by Theorem~\ref{THM:normal},
$f\notin \im(\Delta_K)$.
\end{example}

\subsection{Special reduction}\label{SUBSEC:specialred}
We consider in this subsection the special part of a rational function in $\set F(y)$,
which, by Lemma~\ref{LEM:special}, will be always zero in the usual shift case. 
Thus we merely need to address the $q$-shift case, in which the special part 
is a Laurent polynomial in $\set F[y,y^{-1}]$, again by Lemma~\ref{LEM:special}.
For this purpose, we require the notion of $\sigma_y$-standard rational functions.
\begin{definition}\label{DEF:standard}
A rational function in~$\set F(y)$ with numerator $u$ and denominator $v$ is said to be {\em $\sigma_y$-standard}
if it is $\sigma_y$-reduced, and additionally in the $q$-shift case, $u(0)q^\ell-v(0)\neq 0$ for any negative integer~$\ell$.
\end{definition}

\begin{theorem}\label{THM:special}
Assume that $\sigma_y$ is the $q$-shift operator.
Let $K\in \set F(y)$ be $\sigma_y$-standard with denominator $v$, and
let $f\in\set F(y)$ be a proper rational function whose denominator is $\sigma_y$-special. 
Then there exist $g\in \set F[y^{-1}]$ 
and $b\in \set F[y]$ such that 
\begin{equation}\label{EQ:specialdecomp}
f = \Delta_K(g) + \frac{b}v.
\end{equation}
\end{theorem}
\begin{proof}
Since $\sigma_y$ is the $q$-shift operator and $f$ is proper with $\sigma_y$-special denominator,
it follows from Lemma~\ref{LEM:special} that $f = a/y^k$ for some $k\in \set N$ and $a\in \set F[y]$ 
with $\deg_y(a)<k$. Thus $vf = va/y^k$ is a Laurent polynomial in $\set F[y,y^{-1}]$.
Let $m = \tdeg(vf)$. If $m \geq 0$, then $vf\in \set F[y]$ and thus
letting $g = 0$ and $b = vf$ concludes the proof. Now assume that $m<0$, 
and write $K = u/v$, where $u\in \set F[y]$ with $\gcd(u,v) = 1$.
Since $K$ is $\sigma_y$-standard, 
$u(0)q^m-v(0) \neq 0$. Define
\[
g_0 = \frac{cy^m}{u(0)q^m-v(0)},
\]
where $c$ is the coefficient of $y^m$ in~$vf$. Then $g_0\in \set F[y^{-1}]$ and
\[
vf - (u\sigma_y(g_0) - vg_0) 
= vf - (cy^m + \text{higher terms in $y$}),
\]
which is again a Laurent polynomial in $\set F[y,y^{-1}]$ but with tail degree greater than~$m$,
as opposed to $m$ for the initial Laurent polynomial $vf$.
Thus, repeating this process at most $(-m)$ times yields a Laurent polynomial in $\set F[y,y^{-1}]$ 
with tail degree at least $0$, that is, a polynomial in $\set F[y]$. In other words,
we can find $g \in \set F[y^{-1}]$ and $b\in \set F[y]$ such that 
$vf = u\sigma_y(g)-vg+b$, giving
\[
f = \frac{vf}v = \frac{u\sigma_y(g)-vg + b}v
=\Delta_K(g) + \frac{b}v.
\]
\end{proof}
We now turn the above proof into the following algorithm.

\smallskip\noindent
{\bf SpecialReduction.} Assume that $\sigma_y$ is the $q$-shift operator.
Given a $\sigma_y$-standard rational function $K\in \set F(y)$ 
with numerator $u$ and denominator $v$,
and a proper rational function $f\in\set F(y)$ whose denominator is $\sigma_y$-special, 
compute a Laurent polynomial $g\in \set F[y^{-1}]$ and a polynomial $b\in \set F[y]$ 
such that \eqref{EQ:specialdecomp} holds.
\begin{enumerate}
\item Set $g=0$, $b = vf$ and $m = \tdeg_y(b)$. 
\item While $m <0$ do
\begin{itemize}
\item[2.1] Set $g_0 = cy^m/(u(0)q^m-v(0))$, where $c$ is the coefficient of $y^m$ in~$b$.
\item[2.2] Update $g$ to be $g+g_0$, $b$ to be $b-(u\sigma_y(g_0)-vg_0)$, 
and $m$ to be $\tdeg_y(b)$.
\end{itemize}
\item Return $g$ and $b$.
\end{enumerate}
\begin{example}\label{EX:specialred}
Assume that $\sigma_y$ is the $q$-shift operator. Let $K = -qy+1$
and 
\[
f = \frac{(q-1)(q^2-1)}{y^2}, 
\]
whose denominator is $\sigma_y$-special. 
It is readily seen from definition that $K$ is $\sigma_y$-standard.
Applying the special reduction to $f$ with respect to $K$ yields
\[
f = \Delta_K\Big(\frac{q^2(y-q+1)}{y^2}\Big) + \frac{q^2}{v},
\]
where $v = 1$.
\end{example}
\subsection{Polynomial reduction}\label{SUBSEC:polyred}
To deal with the remaining polynomial part, we present in this subsection a polynomial reduction which 
reduces the input polynomial into one lying in a finite-dimensional linear subspace over $\set F$. 
This reduction was first presented in \cite{BCCLX2013},
and later extended in various ways~\cite{CHKL2015,CKK2016,CvHKK2018,DHL2018,CDK2021}.

Let $K\in \set F(y)$ be $\sigma_y$-standard with numerator $u$ and denominator $v$. 
We define an $\set F$-linear map $\phi_K$
from $\set F[y]$ to itself by sending $p$ to $u\sigma_y(p)-vp$ for all $p\in \set F[y]$,
and call it the {\em map for polynomial reduction with respect to~$K$}.
Then the image of $\phi_K$, denoted by $\im(\phi_K)$, is an $\set F$-linear subspace of~$\set F[y]$.
We denote by $\im(\phi_K)^\top$ the $\set F$-linear subspace of $\set F[y]$ spanned by 
monomials in $y$ whose $y$-degrees are distinct from those of all polynomials in $\im(\phi_K)$, that is,
\[
\im(\phi_K)^\top = \spanning_{\set F}\big\{y^d\mid d\in \set N \ \text{and}\ d \neq \deg_y(p)\ 
\text{for all}\ p \in \im(\phi_K)\big\}.
\]
Following the proof of \cite[Lemma 4.1]{CHKL2015} verbatim, we obtain that 
$\set F[y] = \im(\phi_K) \oplus\im(\phi_K)^\top$. 
Thus we call $\im(\phi_K)^\top$ the {\em standard complement of $\im(\phi_K)$}.
Elements in a standard complement enjoy an important property, which will be useful
in determining the $\sigma_y$-summability of $\sigma_y$-hypergeometric terms.
\begin{lemma}\label{LEM:scomp}
Let $K\in \set F(y)$ be $\sigma_y$-standard with denominator $v$.
If $p\in \im(\phi_K)^\top$ and $p/v\in \im(\Delta_K)$, then $p$ is equal to zero.
\end{lemma}
\begin{proof}
Write $K = u/v$, where $u\in \set F[y]$ with $\gcd(u,v) = 1$.
Assume that $p\in \im(\phi_K)^\top$ and $p/v\in \im(\Delta_K)$. Then there exists $g\in \set F(y)$ such that
$p/v = K\sigma_y(g)-g$, or equivalently,
\begin{equation}\label{EQ:scomp}
p = u\sigma_y(g) - v g \in \set F[y].
\end{equation}
It suffices to show that $g$ is a polynomial in $\set F[y]$, because then 
$p\in \im(\phi_K)\cap \im(\phi_K)^\top=\{0\}$ and thus $p = 0$.
Suppose that $g\notin\set F[y]$. Then its denominator $d$ has a monic irreducible factor~$a\in\set F[y]$, 
which is either $\sigma_y$-normal or $\sigma_y$-special.
Assume that $a$ is $\sigma_y$-normal and let $\alpha\in\set N[\sigma_y,\sigma_y^{-1}]$ be the $\sigma_y$-exponent
of~$a$ in $d$ with tail and head degrees $m$ and $n$, respectively. Then $\sigma_y^m(a)$
is a factor of $d$ but not a factor of $\sigma_y(d)$. It follows from \eqref{EQ:scomp}
that $\sigma_y^m(a)$ divides~$v$. Similarly, since $\sigma_y^{n+1}(a)$ is a factor
of $\sigma_y(d)$ but not a factor of $d$, we have $\sigma_y^{n+1}(a)$ divides $u$ by~\eqref{EQ:scomp},
which contradicts with the condition that $K$ is $\sigma_y$-reduced. 
Thus $a$ must be $\sigma_y$-special. 
Since $a$ is irreducible, it does not belong to~$\set F$. Then it follows from Lemma~\ref{LEM:special} that $\sigma_y$ is the 
$q$-shift operator and $a = y$. Thus $d = y^k$ for some $k\in\set Z^+$ and $g = b/y^k$ for some 
$b\in \set F[y]$ with $y\nmid b$. By \eqref{EQ:scomp}, we get
$
y^kp = uq^{-k}\sigma_y(b) - vb.
$
Since $\sigma_y$ is the $q$-shift operator and $k>0$,
letting $y=0$ on the both sides of the above equation yields that $(u(0)q^{-k} - v(0))b(0)=0$.
Since $y\nmid b$, we have $b(0)\neq 0$ and thus $u(0)q^{-k} - v(0) = 0$. 
Note that $k$ is positive. So we have derived a contradiction with the assumption that $K$ is $\sigma_y$-standard.
Therefore, $g\in \set F[y]$ and then $p = 0$.
\end{proof}
Since $\set F[y] = \im(\phi_K) \oplus\im(\phi_K)^\top$, a polynomial $p\in \set F[y]$ can be uniquely decomposed 
as $p = p_1+p_2$ with $p_1\in \im(\phi_K)$ and $p_2\in \im(\phi_K)^\top$. 
We will see shortly that such a decomposition can be easily computed using echelon bases for $\im(\phi_K)$
and $\im(\phi_K)^\top$. By an {\em echelon basis}, we mean an $\set F$-basis in which 
different elements have distinct $y$-degrees. In order to obtain such a basis, we start by
finding an ordinary $\set F$-basis of~$\im(\phi_K)$.
\begin{lemma}\label{LEM:polyred}
Let $K\in \set F(y)$ be $\sigma_y$-standard. 
Then the following assertions hold.
\begin{itemize}
\item[(i)] The map $\phi_K$ is injective if $K$ is unequal to one in the usual shift case or it is 
not a power of $q$ in the $q$-shift case.
\item[(ii)] The set $\{\phi_K(y^i)\mid i\in \set N\}\setminus\{0\}$ is an $\set F$-basis for $\im(\phi_K)$.
\end{itemize}
\end{lemma}
\begin{proof}
Write $K = u/v$ with $u,v\in \set F[y]$ and $\gcd(u,v) = 1$.

\smallskip\noindent
(i) Assume that $K$ is unequal to one in the usual shift case or it is not a power of $q$ in the 
$q$-shift case. Suppose that there exists a nonzero polynomial $p\in \set F[y]$ such that
$\phi_K(p) = 0$. Then $K = \sigma_y(1/p)/(1/p)$, a contradiction with Proposition~\ref{PROP:redprop},
since $K$ is $\sigma_y$-standard and thus $\sigma_y$-reduced.
Therefore, $\phi_K(p) \neq 0$ for all $p\in \set F[y]\setminus\{0\}$ and then the map $\phi_K$ is injective. 

\smallskip\noindent
(ii) Let $\Lambda = \{\phi_K(y^i)\mid i\in \set N\}$.
It suffices to show the assertion when $K$ is one in the usual shift case or $K$ is a power of $q$ in the
$q$-shift case, for, otherwise, by part~(i), the map $\phi_K$
is injective and thus $\Lambda\setminus\{0\} = \Lambda$ is an $\set F$-basis for $\im(\phi_K)$.

In the usual shift case, namely the case when $\sigma_y$ is the usual shift operator, assume that $K = 1$.
Then we can take $u = v = 1$. It follows that $\phi_K(1) = 0$ and
\[
\phi_K(y^i) = u(y+1)^i-vy^i = i y^{i-1} + (\text{lower terms in~$y$}) \neq 0
\quad\text{for all}\ i \in\set Z^+.
\]
Thus $\Lambda\setminus\{0\} = \{\phi_K(y^i)\mid i\in\set Z^+\}$, which is clearly an $\set F$-basis
for $\im(\phi_K)$. 

In the $q$-shift case, namely the case when $\sigma_y$ is the $q$-shift operator, assume that $K$ is a power of $q$. 
Since $K$ is $\sigma_y$-standard, we have $K = q^{-k}$ for some
$k\in \set N$. So we can take $u = q^{-k}$ and $v = 1$. It follows that
$\phi_K(y^k) = 0$ and
\[
\phi_K(y^i) = (q^{-k+i}-1)y^i\neq 0\quad\text{for all}\ i \in\set N\setminus\{k\}.
\]
Thus $\Lambda\setminus\{0\} = \{\phi_K(y^i)\mid i\in \set N\setminus\{k\}\}$, which is again an $\set F$-basis
for $\im(\phi_K)$.

In summary, the set $\{\phi_K(y^i)\mid i\in \set N\}\setminus\{0\}$ is always an $\set F$-basis for $\im(\phi_K)$.
\end{proof}

We now make a case distinction to demonstrate how to construct an echelon
basis for $\im(\phi_K)$, along with one for $\im(\phi_K)^\top$, from the ordinary $\set F$-basis 
$\{\phi_K(y^i)\mid i\in \set N\}\setminus\{0\}$. 
This distinction is slightly different from the one in~\cite[\S 4.2]{CHKL2015} for the usual shift case and 
that in~\cite[\S 5]{DHL2018} for the $q$-shift case, in the sense that it also includes cases related to 
rational $\sigma_y$-hypergeometric terms.

Let $K\in \set F(y)$ be $\sigma_y$-standard with numerator $u$ and denominator $v$. Set
\[
u = \sum_{i=0}^d u_i y^i\quad\text{and}\quad
v = \sum_{i=0}^d v_i y^i,
\]
where $d = \max\{\deg_y(u),\deg_y(v)\}$ and $u_i,v_i\in \set F$ for all $i =0, \dots, d$.
Note that $u_d$ and $v_d$ cannot be both zero.

\smallskip\noindent
{\em Case~1.} $\sigma_y$ is the usual shift operator. Then
\begin{align}
\phi_K(y^i) &= u(y+1)^i - v y^i = u((y+1)^i- y^i) + (u-v)y^i
\nonumber\\[1ex]
&=(u_d-v_d)y^{d+i} + (iu_d +u_{d-1}-v_{d-1})y^{d+i-1} + 
(\text{lower terms in~$y$})\label{EQ:shiftcase}
\end{align}
for all $i\in \set N$.

\smallskip
{\em Case~1.1.} $u_d-v_d\neq 0$. Then by \eqref{EQ:shiftcase}, 
$\deg_y(\phi_K(y^i))=d+i$ for all $i\in \set N$. This implies that the images of different
powers of $y$ under $\phi_K$ have distinct $y$-degrees, and thus form an echelon basis for $\im(\phi_K)$.
It follows that $\im(\phi_K)^\top$ has an echelon basis $\{1,y,\dots,y^{d-1}\}$
and its dimension is equal to $d$.

\smallskip
{\em Case~1.2.} $u_d-v_d = 0$ and $d=0$. Then $u_d = v_d \neq 0$ and $K=1$. 
By \eqref{EQ:shiftcase}, $\phi_K(1) = 0$ and $\deg_y(\phi_K(y^i)) = i-1$ for all $i\in \set Z^+$, 
implying that $\{\phi_K(y^i)\mid i\in\set Z^+\}$ is an echelon basis for $\im(\phi_K)$
and thus $\im(\phi_K)^\top = \{0\}$, that is, $\dim(\im(\phi_K)^\top) = 0$. 

\smallskip
{\em Case~1.3.} $u_d-v_d = 0$, $d>0$ and $iu_d +u_{d-1}-v_{d-1}\neq 0$ for all $i\in \set N$.
Then by \eqref{EQ:shiftcase}, $\deg_y(\phi_K(y^i)) = d+i - 1$ for all $i \in \set N$. 
Similar to Case~1.1, we see that $\{\phi_K(y^i)\mid i\in \set N\}$ is an echelon basis for $\im(\phi_K)$.
It follows that $\im(\phi_K)^\top$ has an echelon basis $\{1,y,\dots,y^{d-2}\}$
and its dimension is equal to $d-1$.

\smallskip
{\em Case~1.4.} $u_d-v_d = 0$, $d>0$ and $ku_d +u_{d-1}-v_{d-1}= 0$ for some $k\in \set N$.
Note that $u_d = v_d \neq 0$. So the integer $k=(v_{d-1}-u_{d-1})/u_d$ is unique. 
Then by~\eqref{EQ:shiftcase}, 
$\deg_y(\phi_K(y^i)) = d+i-1$ for all $i\in \set N$ with $i\neq k$, and $\deg_y(\phi_K(y^k)) < d+k-1$.
Since $d>0$, we have $K\neq 1$. By Lemma~\ref{LEM:polyred}, $\phi_K(y^k)\neq 0$ and 
$\{\phi_K(y^i)\mid i\in \set N\}$ is an $\set F$-basis for $\im(\phi_K)$.
Eliminating $y^{d+k-2}, y^{d+k-3},\dots,y^{d-1}$ from $\phi_K(y^k)$ successively
by the elements $\phi_K(y^{k-1}),\phi_K(y^{k-2}),\dots,\phi_K(y^0)$, we will obtain a polynomial
$r\in \set F[y]$ with $\deg_y(r)<d-1$. Since $\phi_K(y^0),\dots,\phi_K(y^{k-1}),\phi_K(y^k)$
are linearly independent over $\set F$, the polynomial $r$ is nonzero. So
$
\{\phi_K(y^i)\mid i\in \set N\ \text{with}\ i\neq k\}\cup\{r\}
$
is an echelon basis for $\im(\phi_K)$.
It follows that $\im(\phi_K)^\top$ has an echelon basis 
\[
\{1,y,\dots,y^{\deg_y(r)-1},y^{\deg_y(r)+1},\dots,y^{d-2},y^{d+k-1}\}
\]
and its dimension is equal to $d-1$.

\smallskip\noindent
{\em Case~2.} $\sigma_y$ is the $q$-shift operator. Then
\begin{align}
\phi_K(y^i) &= uq^iy^i - v y^i
=(u_dq^i-v_d)y^{d+i} + (u_{d-1}q^i-v_{d-1})y^{d+i-1}+\cdots+(u_0q^i-v_0)y^i
\label{EQ:qshiftcase}
\end{align}
for all $i\in \set N$.

\smallskip
{\em Case~2.1.} $u_dq^i-v_d\neq 0$ for all $i\in \set N$. Then by \eqref{EQ:qshiftcase},
$\phi_K(y^i)\neq 0$ and $\deg_y(\phi_K(y^i)) = d+i$ for all $i\in \set N$. 
Thus $\{\phi_K(y^i)\mid i\in \set N\}$ is an echelon basis for $\im(\phi_K)$. 
It follows that $\im(\phi_K)^\top$ has an echelon basis $\{1,y,\dots,y^{d-1}\}$
and its dimension is equal to $d$.

\smallskip
{\em Case~2.2.} $u_dq^k-v_d= 0$ for some $k\in \set N$ and $d = 0$. Then $u_d q^k = v_d \neq 0$
and $K = q^{-k}$. By~\eqref{EQ:qshiftcase}, $\phi_K(y^k) = 0$ and $\deg_y(\phi_K(y^i)) = i$ for all 
$i\in \set N$ with $i\neq k$. Thus $\{\phi_K(y^i)\mid i\in \set N\setminus\{k\}\}$ is an echelon 
basis for $\im(\phi_K)$. It follows that $\im(\phi_K)^\top$ is a one-dimensional subspace of $\set F[y]$ 
spanned by~$\{y^k\}$.

\smallskip
{\em Case~2.3.} $u_dq^k-v_d= 0$ for some $k\in \set N$ and $d>0$. The integer $k$ is unique,
because $q$ is neither zero nor a root of unity. Then by~\eqref{EQ:qshiftcase}, 
$\deg_y(\phi_K(y^i)) = d+i$ for all $i\in \set N$ with $i\neq k$, and 
$\deg_y(\phi_K(y^k)) < d+k$. Since $d>0$, we know that $K$ cannot be a power of~$q$.
It then follows from Lemma~\ref{LEM:polyred} that $\phi_K(y^k)\neq 0$ and $\{\phi_K(y^i)\mid i\in \set N\}$ 
is an $\set F$-basis for~$\im(\phi_K)$.
Similar to Case~1.4, we can successively eliminate 
$y^{d+k-1}, y^{d+k-2},\dots,y^{d}$ from $\phi_K(y^k)$ by 
$\phi_K(y^{k-1}),\phi_K(y^{k-2}),\dots,\phi_K(y^0)$ and obtain a polynomial 
$r\in \set F[y]$ with $\deg_y(r)<d$. Since $\phi_K(y^0),\dots,\phi_K(y^{k-1}),\phi_K(y^k)$
are linearly independent over $\set F$, we have $r\neq 0$. Thus 
$
\{\phi_K(y^i)\mid i\in \set N\ \text{with}\ i\neq k\}\cup\{r\}
$
is an echelon basis for $\im(\phi_K)$.
It follows that $\im(\phi_K)^\top$ has an echelon basis 
\[
\{1,y,\dots,y^{\deg_y(r)-1},y^{\deg_y(r)+1},\dots,y^{d-1},y^{d+k}\}
\]
and its dimension is equal to $d$. 

\smallskip
The above case distinction leads to an interesting consequence, which 
tells us that the standard complement has a finite dimension, implying that 
all polynomials therein are ``sparse".
\begin{proposition}\label{PROP:sparse}
Let $K\in \set F(y)$ be $\sigma_y$-standard with numerator $u$ and denominator $v$.
The standard complement of $\im(\phi_K)$ is of dimension 
\begin{equation}\label{EQ:dimension}
\max\{\deg_y(u),\deg_y(v)\}
+\begin{cases}
-\llbracket0\leq \deg_y(u-v)\leq \deg_y(u)-1\rrbracket &\text{in the usual shift case,}\\[2ex]
\llbracket K\ \text{is a nonpositive power of $q$}\rrbracket &\text{in the $q$-shift case,}
\end{cases}
\end{equation}
where the notation $\llbracket \cdot \rrbracket$ equals 1 if the argument 
is true and 0 otherwise.
\end{proposition} 
\begin{example}\label{EX:polynomialred}
Assume that $\sigma_y$ is the $q$-shift operator. Let $K = -qy+1$, which is $\sigma_y$-standard.
According to Case~2.1, $\im(\phi_K)$ has an echelon basis $\{\phi_K(y^i)\mid i\in \set N\}$. Thus
$\im(\phi_K)^\top$ has a basis $\{1\}$ and its dimension is one.
\end{example}
With echelon bases at hand, we are able to project 
a polynomial onto $\im(\phi_K)$ and $\im(\phi_K)^\top$, respectively.
The main process is summarized below.

\smallskip\noindent
{\bf PolynomialReduction.} Given a $\sigma_y$-standard rational function $K\in \set F(y)$, 
and a polynomial $b\in \set F[y]$, 
compute two polynomials $a,p\in \set F[y]$ with $p\in \im(\phi_K)^\top$ such that
$b = \phi_K(a) + p$.
\begin{enumerate}
\item If $b = 0$ then set $a = 0$ and $p = 0$; return.
\item Find polynomials $g_1,\dots,g_m\in \set F[y]$ such that
the set $\{\phi_K(g_1),\dots,\phi_K(g_m)\}$ consists of 
all polynomials in an echelon basis of $\im(\phi_K)$ whose 
$y$-degrees are no more than $\deg_y(b)$, and 
\[
0\leq \deg_y(\phi_K(g_1))<\dots<\deg_y(\phi_K(g_m))\leq \deg_y(b).
\]
\item For $i=m,m-1,\dots,1$, perform linear elimination to find $c_m,c_{m-1},\dots,c_1\in \set F$
such that 
\[
b-\sum_{i=1}^mc_i\phi_K(g_i)\in \im(\phi_K)^\top.
\]
\item Set $a = \sum_{i=1}^mc_ig_i$ and $p = b-\sum_{i=1}^mc_i\phi_K(g_i)$, and return.
\end{enumerate}

\subsection{Remainders of rational functions}
Incorporating the normal, the special and the polynomial
reduction, we are able to further reduce the shell to a ``minimal normal form".
The following definition formalizes such a form.
\begin{definition}\label{DEF:remainder}
Let $K\in \set F(y)$ be $\sigma_y$-standard with denominator $v$, and let $f\in\set F(y)$. 
Another rational function
$r$ in $\set F(y)$ is called a {\em $\sigma_y$-remainder} of $f$ with respect to $K$ if 
$f-r\in \im(\Delta_K)$ and $r$ can be written in the form
\begin{equation}\label{EQ:remainder}
r = h + \frac{p}v,
\end{equation}
where $h\in \set F(y)$ is proper with denominator being $\sigma_y$-normal and strongly
coprime with $K$, and $p\in \im(\phi_K)^\top$. For brevity, we just say that $r$
is a $\sigma_y$-remainder with respect to~$K$ if $f$ is clear from the context.
In addition, we call the denominator of $h$ the {\em significant denominator} of~$r$.
\end{definition}
The notion of significant denominators is well-defined, since the denominator of $h$ in \eqref{EQ:remainder}
is strongly coprime with $K$ and thus is coprime with~$v$.

Remainders describe the ``minimum" distance of a rational function in $\set F(y)$ from 
the $\set F$-linear subspace $\im(\Delta_K)$, and help us decide the $\sigma_y$-summability.
\begin{proposition}\label{PROP:remainder}
Let $K\in \set F(y)$ be $\sigma_y$-standard with numerator $u$ and denominator $v$,
and let $r$ be a $\sigma_y$-remainder 
with respect to~$K$ of the form \eqref{EQ:remainder}. Then the following assertions hold.
\begin{itemize}
\item[(i)] The total number of nonzero terms in $p$ is no more than the number given by~\eqref{EQ:dimension}.
\item[(ii)] If there exists $\tilde r\in \set F(y)$ such that $r-\tilde r\in \im(\Delta_K)$, 
then by writing $\tilde r$ in the form
\begin{equation}\label{EQ:tilder}
\tilde r = \tilde h + \frac{\tilde p}v
\end{equation}
for some $\tilde h\in \set F(y)$ and $\tilde p\in \set F[y]$,
we have the $y$-degree of the denominator of $h$ is no more than that of $\tilde h$.
\item[(iii)] $r\in \im(\Delta_K)$ if and only if $r = 0$.
\end{itemize}
\end{proposition}
\begin{proof}
(i) Since $p\in \im(\phi_K)^\top$, the assertion immediately follows by Proposition~\ref{PROP:sparse}.

\smallskip\noindent
(ii) Assume that there exists $\tilde r\in \set F(y)$ such that $r-\tilde r\in \im(\Delta_K)$ and
write $\tilde r$ in the form \eqref{EQ:tilder}. Then by~\eqref{EQ:remainder}, we obtain
\[
h-\tilde h + \frac{p-\tilde p}v \in \im(\Delta_K).
\]
The assertion is thus evident by Lemma~\ref{LEM:minimal}.

\smallskip\noindent
(iii) The sufficiency is clear. For the necessity, assume that $r\in \im(\Delta_K)$, that is, $h+p/v\in \im(\Delta_K)$. 
We see from Theorem~\ref{THM:normal} that $h$ is actually the zero polynomial. It then
follows that $p/v\in \im(\Delta_K)$. Since $K$ is $\sigma_y$-standard and $p\in \im(\phi_K)^\top$, 
we have $p = 0$ by Lemma~\ref{LEM:scomp},
and thus $r = 0$.
\end{proof}

For describing our reduction algorithm, it remains to show that every nonzero rational function 
in~$\set F(y)$ has a $\sigma_y$-standard kernel, and it is not hard to construct one.
\begin{proposition}\label{PROP:kernel}
Let $f$ be a nonzero rational function in~$\set F(y)$. Then $f$ has a $\sigma_y$-standard kernel~$K$.
Moreover, there exists $S\in \set F(y)$ whose denominator does not have any nontrivial $\sigma_y$-special
factors such that $(K,S)$ is an RNF of~$f$.
\end{proposition}
\begin{proof}
The assertions are evident in the usual shift case, because in this case, $\sigma_y$-standard rational functions
in $\set F(y)$ are exactly $\sigma_y$-reduced ones, and all $\sigma_y$-special polynomials in~$\set F[y]$ belong 
to $\set F$ by Lemma~\ref{LEM:special}.

Now assume that $\sigma_y$ is the $q$-shift operator. Let $f\in\set F(y)\setminus\{0\}$ and 
$(\tilde K,\tilde S)$ be an RNF of~$f$.
Then $\tilde K$ is $\sigma_y$-reduced. Notice that $\sigma_y(y^\ell)/y^\ell = q^\ell$ for all $\ell\in\set Z$. 
So we may assume without loss of generality that $\tilde S$ is $\sigma_y$-monic.
Write $\tilde K = u/v$ with $u,v\in \set F[y]$ and $\gcd(u,v) = 1$. Notice that the constant terms $u(0)$ and $v(0)$ cannot be both zero.
If one of them is zero, or neither is zero and meanwhile $u(0)/v(0)$ is not a positive power of~$q$, then $u(0) q^\ell-v(0) \neq 0$ 
for any negative integer~$\ell$ as $q$ is nonzero, implying that $\tilde K$ is already $\sigma_y$-standard 
and $(\tilde K,\tilde S)$ gives 
a desired RNF of~$f$. Suppose that $u(0),v(0)$ are both nonzero and $u(0)/v(0) = q^m$ for some $m\in \set Z^+$.
Define $K= q^{-m}\tilde K$ and $S = y^m \tilde S$. Since $(\tilde K,\tilde S)$ is an RNF of~$f$, it follows that
\[
f = \tilde K\frac{\sigma_y(\tilde S)}{\tilde S} = q^{-m}\tilde K\frac{\sigma_y(y^m\tilde S)}{y^m\tilde S}
= K \frac{\sigma_y(S)}{S}.
\]
Since $\tilde K$ is $\sigma_y$-reduced, so is $K$ and thus $(K,S)$ is an RNF of~$f$.
Since $K= q^{-m}\tilde K$, it has numerator $q^{-m}u$ and denominator~$v$. Notice that $u(0)/v(0) = q^m$
and $q$ is not a root of unity. So $q^{-m}u(0) q^\ell - v(0) = v(0)(q^\ell-1)\neq 0$ for any negative integer~$\ell$.
It follows that $K$ is $\sigma_y$-standard. It remains to show that the denominator of $S$ has no
nontrivial $\sigma_y$-special factors, which, by Lemma~\ref{LEM:special}, is equivalent to verify that $y$ does
not divide the denominator of~$S$. This follows by the observation that $S = y^m \tilde S$, $m>0$ and 
$\tilde S$ is $\sigma_y$-monic.
\end{proof}

With everything in place, we now present our reduction algorithm, 
which determines the $\sigma_y$-summability of a $\sigma_y$-hypergeometric term
without solving any auxiliary difference equations explicitly.

\smallskip\noindent
{\bf HypergeomReduction.} Given a $\sigma_y$-hypergeometric term $T$, compute  
a $\sigma_y$-hypergeometric term $H$ whose $\sigma_y$-quotient $K$ is $\sigma_y$-standard, 
and two rational functions $g,r\in\set F(y)$ such that 
\begin{equation}\label{EQ:hyperred}
T = \Delta_y(gH) + rH,
\end{equation}
and $r$ is a $\sigma_y$-remainder with respect to~$K$.
\begin{enumerate}
\item Find a $\sigma_y$-standard kernel~$K$ and a corresponding shell $S$ of $\sigma_y(T)/T$, and set $H = T/S$.

\item Compute the canonical representation 
\begin{equation}\label{EQ:Screp}
S = f_p+f_s+f_n,
\end{equation}
where $f_p$, $f_s$ and $f_n$ are the polynomial, special and normal parts of~$S$, respectively.

\item Apply {\bf NormalReduction} to $f_n$ with respect to $K$ to find $g,h\in \set F(y)$
and $b\in \set F[y]$ such that
$
f_n = \Delta_K(g) + h + b/v,
$
and $h$ is proper whose denominator is $\sigma_y$-normal and strongly coprime with~$K$.

\item If $f_s\neq 0$ then apply {\bf SpecialReduction} to $f_s$ with respect to $K$ to find
$g_s\in \set F[y^{-1}]$ and $b_s\in \set F[y]$ such that
$
f_s = \Delta_K(g_s) + b_s/v,
$
and update $g$ to be $g+g_s$ and $b$ to be $b+b_s$.
\item Apply {\bf PolynomialReduction} to $vf_p + b$ with respect to $K$ to find 
$a\in \set F[y]$ and $p\in \im(\phi_K)^\top$ such that 
$
vf_p + b = \phi_K(a) + p.
$

\item Update $g$ to be $a+g$ and set $r = h+ p/v$, and return $H,g,r$.
\end{enumerate}
\begin{remark}
If the shell found in step~1 of the above algorithm is further chosen to be one whose denominator contains no nontrivial
$\sigma_y$-special factors (cf.\ Proposition~\ref{PROP:kernel}), then step~4 in the above algorithm
can be completely skipped.
\end{remark}
\begin{theorem}\label{THM:shellred}
For a $\sigma_y$-hypergeometric term $T$, the algorithm {\bf HypergeomReduction}
computes a $\sigma_y$-hypergeometric term $H$ whose $\sigma_y$-quotient $K$ is $\sigma_y$-standard, 
a rational function $g\in\set F(y)$ and a $\sigma_y$-remainder $r$ with respect to~$K$ such that 
\eqref{EQ:hyperred} holds.
Moreover, $T$ is $\sigma_y$-summable if and only if $r= 0$.
\end{theorem}
\begin{proof}
The correctness of step~1 is guaranteed by Proposition~\ref{PROP:kernel}.
Since $S\in\set F(y)$, the canonical representation \eqref{EQ:Screp} holds.
Applying {\bf NormalReduction} to $f_n$ and {\bf SpecialReduction} to $f_s$ if $f_s\neq 0$
with respect to~$K$, we obtain, after step~4, that
\[
f_n+f_s = \Delta_K(g) + h +\frac{b}v,
\]
which, together with \eqref{EQ:Screp}, leads to
\begin{equation}\label{EQ:shell0}
S = \Delta_K(g)+h + \frac{vf_p + b}v.
\end{equation}
The algorithm {\bf PolynomialReduction} in step~5 then computes the decomposition
\[
vf_p + b = \phi_K(a) + p = u\sigma_y(a) - va + p.
\]
Substituting this into \eqref{EQ:shell0}, we see that
\[
S = \Delta_K(g) + h + K\sigma_y(a) - a + \frac{p}v
= \Delta_K(a+g) + h + \frac{p}v.
\]
After renaming the symbols in step~6 and multiplying both sides by~$H$, we get \eqref{EQ:hyperred} since
$T = SH$ and $K=\sigma_y(H)/H$. 
Thus $T$ is $\sigma_y$-summable if and only if $rH$ is $\sigma_y$-summable, 
which happens if and only if $r\in\im(\Delta_K)$ and then, by Proposition~\ref{PROP:remainder} (iii), 
is equivalent to say that $r = 0$.
\end{proof}
\begin{example}\label{EX:shiftshellred}
Assume that $\sigma_y$ is the usual shift operator. Let
\[
T = \frac{y^3 + 4y^2 + 2y - 2}{(y + 1)(y + 2)}y!.
\]
Clearly, $T$ is a hypergeometric term whose $\sigma_y$-quotient has a kernel $K = y+1$ and 
a corresponding shell
\[
S = \frac{y^3 + 4y^2 + 2y - 2}{(y + 1)(y + 2)}.
\]
Since $K$ is $\sigma_y$-reduced, it is $\sigma_y$-standard by definition. Note that in the usual 
shift case, all $\sigma_y$-special polynomials in $\set F[y]$ belong to $\set F$. So we have obtained 
a $\sigma_y$-standard kernel $K$ and a corresponding shell $S$ whose denominator $(y+1)(y+2)$ 
has no nontrivial $\sigma_y$-special factors. Computing the canonical representation of $S$ gives
\[
S \ \ = \ \ \underbrace{y+1\vphantom{-\frac{3y+4}{(y+1)(y+2)}}}_{f_p}\ \
+\underbrace{0\vphantom{-\frac{3y+4}{(y+1)(y+2)}}}_{f_s}
+\ \ \underbrace{\Big(-\frac{3y+4}{(y+1)(y+2)}\Big)}_{f_n}.
\]
For the normal part $f_n$, its denominator admits a strong $\sigma_y$-factorization
$(y+1)(y+2) = d^\alpha$ with $d = y+2$ and $\alpha = \sigma_y^{-1}+1$. Applying the normal reduction to $f_n$
with respect to $K$ yields
\[
-\frac{3y+4}{(y+1)(y+2)} = \Delta_K\Big(\frac1{y+1}\Big) - \frac1{y+2} - \frac1v,
\]
where $v = 1$. Since $f_s = 0$, we skip the special reduction. Combining with the polynomial part~$f_p$, 
we then use the polynomial reduction with respect to $K$ to obtain $v (y+1) -1 = \phi_K(1) + 0$. Thus
\[
S = \Delta_K\Big(\frac{y+2}{y+1}\Big) -\frac1{y+2},
\]
and consequently,
\[
T = \Delta_y\Big(\frac{y+2}{y+1}H\Big)-\frac1{y+2}H,
\]
where $H = T/S = y!$. So the term $T$ is not $\sigma_y$-summable.
\end{example}

\begin{example}\label{EX:shellred}
Assume that $\sigma_y$ is the $q$-shift operator. Let 
\[
T(k) = \frac{q^{k}(q^{2k+3} - q^{k+2} - q^{k+1}-q^2 + q + 1)}{(q^{k+1} - 1)(q^{k+2} - 1)}(q, q)_k.
\]
Clearly, $T$ is a $q$-hypergeometric term with $q^k = y$ and $\sigma_y(T(k)) =T(k+1)$.
Then the $\sigma_y$-quotient of $T$ has a kernel $\tilde K = -q(qy - 1)$ and a corresponding shell
\[
\tilde S = \frac{q^3y^2 - q^2y - qy - q^2 + q + 1}{(qy - 1)(q^2y - 1)}.
\]
Notice that $\tilde K$ is not $\sigma_y$-standard by definition.
Performing the standardization process as in the proof of Proposition~\ref{PROP:kernel}, 
we obtain a $\sigma_y$-standard kernel $K = -qy+1$ and a corresponding shell $S=y\tilde S$,
whose denominator $(qy-1)(q^2y-1)$ has no nontrivial $\sigma_y$-special factors.
Computing the canonical representation of~$S$ gives
\[
S = \underbrace{y\vphantom{-\frac{q(q-1)y}{(qy - 1)(q^2y - 1)}}}_{f_p}
+\underbrace{0\vphantom{-\frac{q(q-1)y}{(qy - 1)(q^2y - 1)}}}_{f_s}
+\ \  \underbrace{\Big(-\frac{q(q-1)y}{(qy - 1)(q^2y - 1)}\Big)}_{f_n}.
\]
Applying the normal reduction as in Example~\ref{EX:normalred}, we decompose the normal part $f_n$ as
\[
-\frac{q(q-1)y}{(qy - 1)(q^2y - 1)} = \Delta_K\Big(\frac1{qy-1}\Big) + \frac{1}{q(q^2y-1)}+\frac{1/q}{v},
\]
where $v = 1$. Since $f_s=0$, we skip the special reduction.
Combining with the polynomial part $f_p$, we then use the polynomial reduction with 
respect to $K$ to obtain $v\cdot y + 1/q = \phi_K(-1/q) + 1/q$.
Thus
\[
S = \Delta_K\Big(-\frac{qy - q - 1}{q(qy - 1)}\Big)+\frac{qy}{q^2y - 1},
\]
and consequently,
\[
T = \Delta_y\Big(-\frac{q^{k+1} - q - 1}{q(q^{k+1} - 1)}H\Big) + \frac{q^{k+1}}{q^{k+2} - 1}H,
\]
where $H =T/S= (q,q)_k$. So the term $T$ is not $\sigma_y$-summable.
\end{example}

\section{Sum of $\sigma_y$-remainders}\label{SEC:rem}
In order to compute telescopers for $\sigma_y$-hypergeometric terms, we want to parameterize
the reduction algorithm developed in the preceding section so as to reduce the problem to 
determining the linear dependency among certain $\sigma_y$-remainders.  
However, we are confronted with the same difficulty as mentioned in \cite[\S 5]{CHKL2015}
that the sum of two $\sigma_y$-remainders is not necessarily a $\sigma_y$-remainder.
A complete obstruction preventing the linearity from being true is that
the least common multiple of two $\sigma_y$-normal 
polynomials may not again be $\sigma_y$-normal.
The idea used in~\cite[\S 5]{CHKL2015} to circumvent this obstruction can be literally extended
to our general setting. For the sake of completeness and the convenience of later use, we
present the idea in this section with some subtle adjustments.

Let $d$ and $e$ be two nonzero $\sigma_y$-normal polynomials in $\set F[y]$.
By polynomial factorization and dispersion computation (see \cite{AbPe2002b}), 
one can decompose
\begin{equation}\label{EQ:scdecomp}
e = \tilde e\, d_1^{k_1\sigma_y^{\ell_1}} \cdots\, d_m^{k_m\sigma_y^{\ell_m}},
\end{equation}
where $\tilde e\in\set F[y]$ is $\sigma_y$-coprime with $d$, $d_1, \dots, d_m\in\set F[y]$ 
are pairwise distinct and $\sigma_y$-monic irreducible factors of $d$, $\ell_1,\dots,\ell_m$ are nonzero integers, 
and $k_1,\dots,k_m\in\set Z^+$ are multiplicities of the factors 
$\sigma_y^{\ell_1}(d_1), \dots, \sigma_y^{\ell_m}(d_m)$ in $e$, respectively.
Note that such a decomposition is unique up to the order of factors. 
Following \cite{CHKL2015}, we refer to \eqref{EQ:scdecomp} as the {\em $\sigma_y$-coprime decomposition}
of $e$ with respect to~$d$.
\begin{theorem}\label{THM:remsum}
Let $K\in \set F(y)$ be $\sigma_y$-standard,
and let $r,s$ be two $\sigma_y$-remainders
with respect to~$K$. Then there exists a $\sigma_y$-remainder $t$ with respect to~$K$
such that $s-t\in \im(\Delta_K)$ and $\lambda r + \mu t$ for all $\lambda,\mu \in \set F$
is a $\sigma_y$-remainder with respect to~$K$.
\end{theorem}
\begin{proof}
Let $d$ and $e$ be significant denominators of $r$ and $s$, respectively. 
Then $d$ and $e$ are both $\sigma_y$-normal and strongly coprime with~$K$.
Let \eqref{EQ:scdecomp} be 
the $\sigma_y$-coprime decomposition of $e$ with respect to~$d$. Since $d$ and $e$
are both $\sigma_y$-normal, the factors $\tilde e, \sigma_y^{\ell_1}(d_1)$, $\dots, 
\sigma_y^{\ell_m}(d_m)$ are pairwise coprime. Then $s$ can be decomposed as
\begin{equation}\label{EQ:sform}
s = \sum_{i=1}^ms_i+\frac{\tilde a}{\tilde e}+\frac{b}v,
\end{equation}
where $s_i\in \set F(y)$ is a nonzero proper rational function with denominator 
$d_i^{k_i\sigma_y^{\ell_i}}$ for $i=1,\dots,m$, $\tilde a\in \set F[y]$ 
with $\deg_y(\tilde a)<\deg_y(\tilde e)$, $b\in \im(\phi_K)^\top$ and $v$ is the denominator of~$K$. 
For each $i=1,\dots,m$, applying Lemma~\ref{LEM:normalmonomial} to $s_i$ delivers 
\begin{equation}\label{EQ:si}
s_i = \Delta_K(g_i) + \frac{a_i}{d_i^{k_i}} + \frac{b_i}v,
\end{equation}
where $g_i\in \set F(y)$, $a_i,b_i\in \set F[y]$ and $\deg_y(a_i)<k_i\deg_y(d_i)$.
It then follows from \eqref{EQ:sform} that 
\[
s = \Delta_K\Big(\sum_{i=1}^mg_i\Big) 
+ \sum_{i=1}^m\frac{a_i}{d_i^{k_i}} + \frac{\tilde a}{\tilde e}
+ \frac{b+\sum_{i=1}^mb_i}v.
\]
By polynomial reduction, we can find $a\in \set F[y]$ and $p\in \im(\phi_K)^\top$
so that $\sum_{i=1}^mb_i = \phi_K(a) + p$. Let 
\[
h = \sum_{i=1}^m\frac{a_i}{d_i^{k_i}} + \frac{\tilde a}{\tilde e}
\quad\text{and}\quad
t = h+ \frac{b+p}v.
\]
Then $s = \Delta_K(\sum_{i=1}^m g_i+a) + t$ and thus $s-t \in \im(\Delta_K)$. 

Notice that the denominator of $h$ divides the polynomial 
$
\tilde e\, d_1^{k_1}\cdots\, d_m^{k_m},
$
which is $\sigma_y$-normal and strongly coprime with~$K$ since both $d$ and $e$ are 
$\sigma_y$-normal and strongly coprime with~$K$,
and is $\sigma_y$-coprime with $d$ since $\tilde e$ is $\sigma_y$-coprime
with $d$ and $d_i\mid d$ for all $i=1,\dots,m$.
Thus $h$ is proper whose denominator is $\sigma_y$-normal, strongly coprime with~$K$
and $\sigma_y$-coprime with~$d$. Since $b+p\in \im(\phi_K)^\top$, we conclude 
that $t$ is a $\sigma_y$-remainder with respect to $K$ whose significant denominator is 
$\sigma_y$-coprime with~$d$. It follows that the least common multiple of the significant 
denominators of $t$ and $r$ is $\sigma_y$-normal. Therefore, $\lambda r + \mu t$ for all 
$\lambda,\mu \in \set F$ is a $\sigma_y$-remainder with respect to~$K$.
\end{proof}

The proof of the above theorem contains an algorithm, which is outlined below.

\smallskip\noindent
{\bf RemainderLinearization.} Given a $\sigma_y$-standard rational function $K\in \set F(y)$, and
two $\sigma_y$-remainders $r,s$ with respect to~$K$, compute a rational function $g\in \set F(y)$ and
another $\sigma_y$-remainder $t$ with respect to~$K$ such that
\[
s = \Delta_K(g) + t
\]
and $\lambda r + \mu t$ for all $\lambda,\mu \in \set F$ is a $\sigma_y$-remainder 
with respect to~$K$.
\begin{enumerate}
\item Set $d$ and $e$ to be the significant denominators of $r$ and $s$, respectively.

\item Compute the $\sigma_y$-coprime decomposition \eqref{EQ:scdecomp} of~$e$
with respect to~$d$, and then decompose $s$ into the form \eqref{EQ:sform}.

\item For $i=1,\dots, m$, apply Lemma~\ref{LEM:normalmonomial} to $s_i$ to find
$g_i\in \set F(y)$ and $a_i,b_i\in \set F[y]$ with $\deg_y(a_i)<k_i\deg_y(d_i)$ such that
\eqref{EQ:si} holds.

\item Apply {\bf PolynomialReduction} to $\sum_{i=1}^mb_i$ to find $a\in \set F[y]$
and $p\in \im(\phi_K)^\top$ such that 
$
\sum_{i=1}^mb_i = \phi_K(a) + p.
$

\item Set
\[
g = \sum_{i=1}^m g_i + a\quad\text{and}\quad
t = \sum_{i=1}^m\frac{a_i}{d_i^{k_i}} + \frac{\tilde a}{\tilde e}
+ \frac{b+p}v,
\]
and return.
\end{enumerate}

\begin{example}\label{EX:sumofrems}
Assume that $\sigma_y$ is the $q$-shift operator. Let $K = -qy+1$, which is $\sigma_y$-standard,
and let $r = qy/(q^2y - 1)$, $s = q^2y/(q^3y-1)$. Then both $r$ and $s$ are $\sigma_y$-remainders
with respect to~$K$, but their sum is not one, because the denominator $(q^2y-1)(q^3y-1)$ is not
$\sigma_y$-normal. Using the algorithm {\bf RemainderLinearization}, 
we can find another $\sigma_y$-remainder $t=q^3y/(q^2-1)/(q^2y-1)$ of $s$ with respect to~$K$
such that $r+t= q(2q^2 - 1)y/(q^2 - 1)/(q^2y - 1)$, which 
is clearly a $\sigma_y$-remainder with respect to~$K$. 
\end{example}

\section{Creative telescoping via reduction}\label{SEC:rct}
In this section, we will translate terminologies concerning univariate hypergeometric terms to
bivariate ones and propose an algorithm for computing minimal telescopers as well as certificates 
in this bivariate setting.

Let $\set K$ be a field of characteristic zero, and $\set K(x,y)$ be the field of rational functions
in $x$ and $y$ over $\set K$. Let $\sigma_x$ and $\sigma_y$ be both either the usual shift operators 
with respect to $x$ and $y$ respectively defined by
\[
\sigma_x(f(x,y)) = f(x+1,y)\quad\text{and}\quad
\sigma_y(f(x,y)) = f(x,y+1),
\]
or the $q$-shift operators with respect to $x$ and $y$ respectively defined by
\[
\sigma_x(f(x,y)) = f(qx,y)\quad\text{and}\quad
\sigma_y(f(x,y)) = f(x,qy)
\]
for any $f\in \set K(x,y)$, where $q\in \set K$ is neither zero nor a root of unity.
Similarly, we will refer to the former case as the
usual shift case, and the latter one as the $q$-shift case. The pair $(\set K(x,y),\{\sigma_x,\sigma_y\})$
forms a partial ($q$-)difference field. Let $R$ be a {\em partial ($q$-)difference ring extension} of 
$(\set K(x,y),\{\sigma_x,\sigma_y\})$, that is, $R$ is a ring containing $\set K(x,y)$ together
with two distinguished endomorphisms $\sigma_x$ and $\sigma_y$ from $R$ to itself, 
whose respective restrictions to $\set K(x,y)$ agree with the two automorphisms defined earlier.
In analogy with the univariate case in Section~\ref{SEC:prelim}, an element $c\in R$ is
called a {\em constant} if it is invariant under the applications of $\sigma_x$ and~$\sigma_y$.
It is readily seen that all constants in $R$ form a subring of $R$. 
\begin{definition}
An invertible element $T$ of $R$ is called a {\em $(\sigma_x,\sigma_y)$-hypergeometric term} 
if there exist $f,g\in \set K(x,y)$ such that $\sigma_x(T) = fT$ and $\sigma_y(T) = gT$.
We call $f$ and $g$ the $\sigma_x$- and $\sigma_y$-quotients of $T$, respectively.
\end{definition}
In the rest of this paper, let $\set F$ be the field $\set K(x)$ and 
$\set F[S_x]$ be the ring of linear recurrence operators
in~$x$ over $\set F$, in which the commutation rule is that $S_xf = \sigma_x(f)S_x$ for all $f\in \set F$.
The application of an operator $L = \sum_{i=0}^\rho \ell_iS_x^i\in \set F[S_x]$ to a 
$(\sigma_x,\sigma_y)$-hypergeometric term $T$ is defined as
\[
L(T) = \sum_{i=0}^\rho \ell_i\sigma_x^i(T).
\]
\begin{definition}
Let $T$ be a $(\sigma_x,\sigma_y)$-hypergeometric term. 
A nonzero linear recurrence operator $L\in \set F[S_x]$ is called a 
{\em telescoper} for $T$ if there exists a $(\sigma_x,\sigma_y)$-hypergeometric 
term $G$ such that
\[
L(T) = \Delta_y(G).
\] 
We call $G$ a corresponding {\em certificate} of~$L$. The {\em order} of a telescoper 
for~$T$ is defined to be its degree in $S_x$.
\end{definition}

For $(\sigma_x,\sigma_y)$-hypergeometric terms, telescopers do not always exist.
Existence criteria were provided by Abramov \cite{Abra2003} for the usual shift case and 
by Chen et al.~\cite{CHM2005} for the $q$-shift case. In order to describe them 
concisely, we introduce the notion of integer-linear rational functions.
Note that two polynomials $a,b$ in $\set K[x,y]$ are associates if and only if
$a=c\,b$ for some $c\in \set K$.
\begin{definition}\label{DEF:intlinear}
An irreducible polynomial $p$ in $\set K[x,y]$ is said to be {\em integer-linear} (over $\set K$)
if there exist $m,n\in \set Z$, not both zero, such that $p$ and $\sigma_x^m\sigma_y^n(p)$ are associates.
A polynomial in~$\set K[x,y]$ is said to be {\em integer-linear} (over $\set K$) if all its irreducible 
factors over $\set K$ are integer-linear. A rational function in $\set F(y)$ is said to be {\em integer-linear} 
(over $\set K$) if its denominator and numerator are both integer-linear.
\end{definition}
We refer to \cite{GHLZ2019, GHLZ2021} for algorithms determining the integer-linearity of rational functions.
The following proposition provides an easy-to-use equivalent form for an integer-linear irreducible polynomial
in~$\set K[x,y]$, which also justifies the ``integer-linear" attribute in the name. 
\begin{proposition}\label{PROP:intlinear}
Let $p$ be an irreducible polynomial in~$\set K[x,y]$.
\begin{itemize}
\item[(i)] In the usual shift case, $p$ is integer-linear if and only if
$
p(x,y) = P(\lambda x+ \mu y)
$
for some $P(z)\in \set K[z]$ and $\lambda,\mu\in \set Z$ not both zero.

\item[(ii)] In the $q$-shift case, $p$ is integer-linear if and only if
$
p(x,y) = x^\alpha y^\beta P(x^\lambda y^\mu)
$
for some $P(z)\in \set K[z]$ and $\alpha,\beta,\lambda,\mu\in \set Z$ with $\lambda, \mu$ not both zero.
\end{itemize}
\end{proposition}
\begin{proof}
(i) We first consider the usual shift case, that is, the case when $\sigma_x$ and $\sigma_y$ are both
usual shift operators. The sufficiency is clear since 
$\sigma_x^{\mu}\sigma_y^{-\lambda}(P(\lambda x+ \mu y)) = P(\lambda x+\mu y)$
for any $P(z)\in \set K[z]$ and $\lambda,\mu\in \set Z$.
Assume that $p$ is integer-linear. Then there are $m,n\in\set Z$, not both zero, such that $p$ and
$\sigma_x^m\sigma_y^n(p)$ are associates. Since $\sigma_x,\sigma_y$ are usual shift operators, 
we have $p(x+m,y+n) = p(x,y)$. The assertion then follows from \cite[Lemma~3.3]{Hou2004}.

\smallskip\noindent
(ii) In the $q$-shift case, that is, in the case when $\sigma_x$ and $\sigma_y$ are both $q$-shift operators,
the sufficiency is clear since
$\sigma_x^{\mu}\sigma_y^{-\lambda}(x^\alpha y^\beta P(x^\lambda y^\mu)) 
= q^{\alpha\mu-\beta\lambda}x^\alpha y^\beta P(x^\lambda y^\mu)$
for any $P(z)\in \set K[z]$ and $\alpha,\beta,\lambda,\mu\in \set Z$.
For the necessity, we assume that $p$ is integer-linear.
Then $p$ is an associate of $\sigma_x^m\sigma_y^n(p)$ for some $m,n\in\set Z$ not both zero. 
Since $p$ is irreducible over~$\set K$ and $\sigma_x,\sigma_y$ are $q$-shift operators,
it follows that $p(q^mx,q^ny) =c\,p(x,y)$ for some $c\in \set K$.
We now adapt the proof of \cite[Lemma~3.3]{Hou2004} into this case.
Let $\ell = \gcd(m,n)$. Since $m,n$ are not both zero, we have $\ell\neq 0$.
Let $\lambda = -n/\ell$ and $\mu = m/\ell$. It is readily seen that $\gcd(\lambda,\mu) = 1$. 
By B{\' e}zout's relation, there exist $s,t\in \set Z$
such that $s\lambda + t \mu = 1$. Define $h(x,y) = p(x^s y^\mu,x^t y^{-\lambda})$.
Then $h\in \set K[x,x^{-1},y,y^{-1}]\subset \set F(y)$ and
\[
h(x, q^\ell y) =p(x^s(q^\ell y)^\mu, x^t(q^\ell y)^{-\lambda})
=p(q^m x^s y^\mu, q^n x^ty^{-\lambda})
= c\,p(x^s y^\mu, x^ty^{-\lambda}) = c\,h(x,y).
\]
Since $\ell \neq 0$, we conclude from Lemma~\ref{LEM:invariant} (ii) that $h/y^k\in \set F$ for some integer~$k$.
Notice that $h\in \set K[x,x^{-1},y,y^{-1}]$. So $h/y^k \in \set K[x,x^{-1}]$ and then
$h(x,y)=x^{-r}y^k P(x)$ for some $r\in \set N$ and $P(z)\in \set K[z]$. Therefore,
\[
p(x,y) = h(x^\lambda y^\mu,x^ty^{-s}) = (x^\lambda y^\mu)^{-r}(x^ty^{-s})^kP(x^\lambda y^\mu)
= x^{-\lambda r+tk}y^{-\mu r- sk}P(x^\lambda y^\mu).
\]
Letting $\alpha= -\lambda r + tk$ and $\beta=-\mu r-sk$ gives $p(x,y) = x^\alpha y^\beta P(x^\lambda y^\mu)$.
\end{proof}

Combining \cite[Theorem~10]{Abra2003} and \cite[Theorem~4.6]{CHM2005}, 
we have the following existence criteria for telescopers.
\begin{theorem}\label{THM:existence}
Let $T$ be a $(\sigma_x,\sigma_y)$-hypergeometric term 
whose $\sigma_y$-quotient has a $\sigma_y$-standard kernel $K$ and a corresponding shell $S$, and let $r$ be a 
$\sigma_y$-remainder of $S$ with respect to $K$.
Then $T$ has a telescoper if and only if the significant  denominator of $r$ 
is integer-linear.
\end{theorem}

When telescopers exist, we are then able to use the reduction algorithm developed
in Section~\ref{SEC:red} to construct a telescoper for a given 
$(\sigma_x,\sigma_y)$-hypergeometric term effectively.

\smallskip\noindent
{\bf HypergeomTelescoping.} Given a $(\sigma_x,\sigma_y)$-hypergeometric term $T$,
compute a telescoper of minimal order for $T$ and its certificate if $T$ has telescopers.

\begin{enumerate}
\item Apply {\bf HypergeomReduction} to $T$ with respect to $y$ to find 
a $(\sigma_x,\sigma_y)$-hypergeometric term $H$ whose $\sigma_y$-quotient $K$ is $\sigma_y$-standard, 
and two rational functions $g_0,r_0\in\set F(y)$ such that 
\begin{equation}\label{EQ:hyperred0}
T = \Delta_y(g_0H) + r_0H,
\end{equation}
and $r_0$ is a $\sigma_y$-remainder with respect to~$K$.
If $r_0=0$ then return $(1,g_0H)$.

\item If the significant denominator of $r_0$ is not integer-linear, then return
``No telescoper exists!".

\item Set $N = \sigma_x(H)/H$ and $r = \ell_0r_0$, where $\ell_0$ is an indeterminate.

For $i=1,2,\dots$ do
\begin{itemize}
\item[3.1] Apply {\bf HypergeomReduction} to $\sigma_x(r_{i-1})NH$ with respect to $y$,
with the choice of the $\sigma_y$-standard kernel $K$ and the shell $\sigma_x(r_{i-1})N$ in its step~1, to find
$\tilde g_i\in \set F(y)$ and a $\sigma_y$-remainder $\tilde r_i$ with respect to $K$ 
such that
\begin{equation}\label{EQ:shelli-1}
\sigma_x(r_{i-1})NH = \Delta_y(\tilde g_iH) + \tilde r_iH.
\end{equation}

\item[3.2] Apply {\bf RemainderLinearization} to $\tilde r_i$ with respect to $r$ and $K$ to
find $\bar g_i\in \set F(y)$ and another $\sigma_y$-remainder $r_i$ with respect to $K$ 
such that
\begin{equation}\label{EQ:remsumi-1}
\tilde r_i = \Delta_K(\bar g_i) + r_i
\end{equation}
and $r+\ell_ir_i$ is a $\sigma_y$-remainder with respect to $K$, where $\ell_i$ is an 
indeterminate.

\item[3.3] Set $g_i = \sigma_x(g_{i-1})N+\tilde g_i + \bar g_i$ and update $r$ to be $r+\ell_i r_i$.

\item[3.4] Find $\ell_0,\dots,\ell_i\in \set F$ such that $r = 0$ by solving a linear system in 
$\ell_0,\dots,\ell_i$ over $\set F$. If there is a nontrivial solution, return
$
(\sum_{j=0}^i\ell_jS_x^j,\ \sum_{j=0}^i\ell_j g_j H).
$
\end{itemize}
\end{enumerate}

\begin{remark}
The algorithm {\bf HypergeomTelescoping} separates the computation of telescopers from that of certificates.
In applications where certificates are irrelevant, we can drop all computations related to the
preimages of $\Delta_y$. In particular, all rational functions $g_i$ can be discarded and we do
not need to calculate $\sum_{j=0}^i\ell_jg_jH$ in the end.
\end{remark}

\begin{theorem}\label{THM:rct}
For a $(\sigma_x,\sigma_y)$-hypergeometric term $T$, 
the algorithm {\bf HypergeomTelescoping} terminates and correctly finds a 
telescoper of minimal order for $T$ and a corresponding certificate when such telescopers exist.
\end{theorem}
\begin{proof}
By Theorem~\ref{THM:shellred}, $r_0=0$ implies that $T$ is $\sigma_y$-summable and thus
$1$ is a telescoper of minimal order for $T$. Together with Theorem~\ref{THM:existence},
we see that steps~1 and 2 are correct.

Now assume that the algorithm proceeds to step~3. Then $T$ has a telescoper of order at least one.
It follows from \eqref{EQ:hyperred0} and $\sigma_x(H) = NH$ that 
\[
\sigma_x(T) = \Delta_y(\sigma_x(g_0)NH) + \sigma_x(r_0)NH.
\]
By viewing $\sigma_x(r_0)NH$ as a $(\sigma_x,\sigma_y)$-hypergeometric term
whose $\sigma_y$-quotient has a $\sigma_y$-standard kernel $K$ and a corresponding shell $\sigma_x(r_0)N$,
we can perform {\bf HypergeomReduction} to $\sigma_x(r_0)NH$ to obtain $\tilde g_1\in \set F(y)$
and a $\sigma_y$-remainder $\tilde r_1$ with respect to $K$ so that 
\eqref{EQ:shelli-1} holds for $i=1$. 
According to Theorem~\ref{THM:remsum}, the algorithm {\bf RemainderLinearization} enables us
to find $\bar g_1\in \set F(y)$ and another $\sigma_y$-remainder $r_1$ with respect to~$K$ 
such that \eqref{EQ:remsumi-1} holds for $i=1$ and $r+\ell_1r_1=\ell_0r_0+\ell_1r_1$ for all 
$\ell_0,\ell_1\in\set F$ is again a $\sigma_y$-remainder with respect to $K$. 
Setting $g_1 = \sigma_x(g_0)N+\tilde g_1+\bar g_1$, we thus get
$
\sigma_x(T) = \Delta_y(g_1H) + r_1H.
$
By a direct induction on $i\in\set Z^+$, we see that 
\begin{equation}\label{EQ:hyperredi}
\sigma_x^i(T) = \Delta_y(g_iH) + r_iH
\end{equation}
holds in the loop of step~3 every time the algorithm passes through step~3.3.
Moreover, $r = \sum_{j=0}^i\ell_j r_j$ for all $\ell_j\in\set F$ is a $\sigma_y$-remainder with respect to~$K$.

Let $\rho\in \set Z^+$ and define $L= \sum_{i=0}^\rho c_iS_x^i$ with $c_i\in \set F$ and $c_\rho \neq 0$. 
Then by \eqref{EQ:hyperredi},
\[
L(T) = \Delta_y\Big(\sum_{i=0}^\rho c_ig_i H\Big) + \Big(\sum_{i=0}^\rho c_i r_i\Big) H.
\]
Since $\sum_{i=0}^\rho c_i r_i$ is a $\sigma_y$-remainder with respect to $K$, we conclude
from Theorem~\ref{THM:shellred} that $L$ is a telescoper for~$T$ if and only if 
$\sum_{i=0}^\rho c_i r_i$ is equal to zero, which happens if and only if
the linear homogeneous system in $\ell_0,\dots,\ell_\rho$ over $\set F$ obtained by equating 
$\sum_{i=0}^\rho \ell_i r_i$ to zero has a nontrivial solution $\ell_0=c_0,\dots,\ell_\rho=c_\rho$. 
Therefore, the first linear dependency among the $r_i$ gives rise to a telescoper of minimal order.
\end{proof}
\begin{example}\label{EX:shiftrct}
Assume that $\sigma_x$ and $\sigma_y$ are both the usual shift operators. Let $T(x,y) = {x+2y\choose y}$. 
Then $T$ is a $(\sigma_x,\sigma_y)$-hypergeometric term whose $\sigma_y$-quotient has a $\sigma_y$-standard
kernel $K = (x + 2y + 1)(x + 2y + 2)/((y + 1)(x + y + 1))$ and a corresponding shell $S = 1$. So $H = T/S= T$.
Applying the algorithm {\bf HypergeomTelescoping} to $T$, we obtain in step~3 that
\[
\sigma_x^i(T) = \Delta_y(g_iH) + \frac{p_i}vH\quad \text{for}\ i = 0,1,2,
\]
where $v = (y + 1)(x + y + 1)$, $p_0 = (-x^2+x+2 y+2)/3$, $p_1 = (-2x^2 - 3xy - x + y + 1)/3$,
$p_2 = (-x^2 - 3xy - 2x - y - 1)/3$, and $g_i\in \set F(y)$ are not displayed here to keep things neat.
By finding an $\set F$-linear dependency among $p_0,p_1,p_2$, we get 
\[
L = S_x^2 - S_x + 1
\]
is a telescoper of minimal order for $T$.
\end{example}
\begin{example}\label{EX:rct}
Assume that $\sigma_x$ and $\sigma_y$ are both the $q$-shift operators. 
Let $T(n,k) = \qbinom{n}{k}_q$. Then $T$ is a $(\sigma_x,\sigma_y)$-hypergeometric term 
with $\sigma_x(T(n,k)) = T(n+1,k)$, $\sigma_y(T(n,k)) = T(n,k+1)$ and $q^n = x$, $q^k = y$. 
The $\sigma_y$-quotient of~$T$ has a $\sigma_y$-standard kernel $K = (x-y)/(y(qy-1))$ and 
a corresponding shell $S = 1$. So $H = T/S = T$. Applying the algorithm {\bf HypergeomTelescoping} to 
$T$, we obtain in step~3 that 
\[
\sigma_x^i(T) = \Delta_y(g_iH) + \frac{p_i}vH\quad\text{for}\ i = 0,1,2,
\]
where $v = y(qy-1)$, $p_0 = x-y$, $p_1 = (qx-1)y$, $p_2 = (qx-1)(x+y)$, and $g_i\in \set F(y)$ are not 
displayed here to keep things neat. By finding an $\set F$-linear dependency among $p_0,p_1,p_2$, 
we get
\[
L = S_x^2-2S_x+1-q^{n+1}
\]
is a telescoper of minimal order for~$T$.
\end{example}

\section{Implementation and applications}\label{SEC:exp}
We have implemented our algorithms in the computer algebra system {\sc Maple 2020}.
The code and its demo are available via the link:
\[\text{\url{http://www.mmrc.iss.ac.cn/~schen/CDGHL2025.html}}\]
We aim to compare their runtime and memory
requirements with the performance of known algorithms so as to get an idea about the efficiency.
Since such experiments for the shift case have been well conducted in~\cite{CHKL2015}, 
we will focus on the $q$-shift case in this  section. All timings are measured in seconds on 
a macOS computer with 32GB RAM and 2.3 GHz Quad-Core Intel Core i7 processors. 
The computations for the experiments did not use any parallelism.

We consider bivariate $q$-hypergeometric terms of the form
\[
T = \frac{f(q^n,q^k)}{g(q^{n+k})}\frac{(q;q)_{2\alpha n+k}}{(q;q)_{n+\alpha k}},
\]
where $f\in \set Q(q)[q^n,q^k]$ is of total degree $1$, $g = p\sigma_z^{\lambda}(p)\sigma_z^\mu(p)$ with 
$p\in \set Q[z]$ of degree $d$ and $\alpha,\lambda,\mu\in\set N$. For a selection of random terms 
of this type for different choices of $(d,\alpha,\lambda,\mu)$, Table~\ref{TAB:timing} compares
the timings of Maple's implementation of $q$-Zeilberger's algorithm (Z) and two variants of the 
algorithm {\bf HypergeomTelescoping} in the $q$-shift case from Section~\ref{SEC:rct}: 
For the column HTC we computed both the telescoper and the certificate, and for HT only the 
telescoper was returned. The difference between these two variants mainly comes from the time 
needed to bring the rational part $r$ in the certificate $rH$ on a common denominator. When it 
is acceptable to keep the rational part as an unnormalized linear combination of rational functions, 
the time is virtually the same as that for HT.

\begin{table}[!ht]
\centering
\begin{tabular}{l|r|r|r|c}
$(d,\alpha,\lambda,\mu)$ & Z & HTC & HT & order\\\hline
(1, 1, 1, 5) & 1.05 & .94 & .50 & 2 \\
(1, 2, 1, 5) & 5082.63 & 1854.00 & 797.41 & 5 \\
(2, 1, 1, 5) & 17.36 & 7.37 & 3.15 & 3 \\
(2, 2, 1, 5) & 167299.50 & 18774.25 & 9778.94 & 6 \\
(3, 1, 1, 5) & 884.95 & 72.20 & 22.98 & 4 \\
(1, 1, 5, 10) & 17.40 & 8.44 & 2.53 & 2 \\
(1, 2, 5, 10) & 39997.13 & 18542.85 & 9139.08 & 5 \\
(1, 1, 10, 15) & 60.93 & 48.97 & 7.92 & 2 \\\hline
\end{tabular}
\caption{\small Comparison of $q$-Zeilberger's algorithm to reduction-based creative telescoping
with and without construction of a certificate for a collection of random terms}
\label{TAB:timing}
\end{table}

Our implementation enhances the applicability of ($q$-)Zeilberger’s algorithm, as illustrated by the following example.
\begin{example}
Stanton~\cite{Stan1995} conjectured the following identity
\begin{equation}\label{EQ:stanton}
\sum_k (-1)^kq^{4k^2}\qbinom{2n}{n-4k}_q=\sum_kq^{2k^2}\qbinom{n}{2k}_q(-q;q^2)_{n-2k}(-1;q^4)_k,
\end{equation}
which was proved by Paule and Riese~\cite{PaRi1997} using $q$-Zeilberger's algorithm. 
They observed that taking the summation range to be $\{\sf k,-\text{\sf Infinity}, \text{\sf Infinity}\}$ saves computing
time. Actually, due to the natural boundary of Gaussian binomial coefficients, we can further omit the computation
of certificates and write down recurrence equations for both sides of~\eqref{EQ:stanton} merely using telescopers.
For the summand in either side of~\eqref{EQ:stanton}, without computing the certificate, our implementation returns
the same telescoper
\[
S_x^3-(q^{2n+5}+q^{2n+4}+q^{2n+3}+1)S_x^2+q^{2n+4}(q^{2n+3}+q^{2n+2}+q^{2n+1}-1)S_x
-q^{2n+3}(q^{2n+2}-1)(q^{2n+1}-1),
\]
where $q^n = x$, thus implying that both sides of~\eqref{EQ:stanton} satisfy the same recurrence equation
\[
\begin{array}{l}
f(n+3)-(q^{2n+5}+q^{2n+4}+q^{2n+3}+1)f(n+2)+q^{2n+4}(q^{2n+3}+q^{2n+2}+q^{2n+1}-1)f(n+1)\\[1ex]
-q^{2n+3}(q^{2n+2}-1)(q^{2n+1}-1)f(n) = 0.
\end{array}
\]
Checking the identity at the initial values $n=0,1,2$ completes the proof of~\eqref{EQ:stanton}.
\end{example}

\section*{Acknowledgments}
We would like to thank the anonymous referee for many useful and constructive suggestions.

\bibliographystyle{plain}

\begin{thebibliography}{10}

\bibitem{Abra2003}
Sergei~A. Abramov.
\newblock When does {Z}eilberger's algorithm succeed?
\newblock {\em Adv. in Appl. Math.}, 30(3):424--441, 2003.

\bibitem{APP1998}
Sergei~A. Abramov, Peter Paule, and Marko Petkov{\v{s}}ek.
\newblock {$q$}-{H}ypergeometric solutions of {$q$}-difference equations.
\newblock {\em Discrete Math.}, 180(1-3):3--22, 1998.

\bibitem{AbPe2001a}
Sergei~A. Abramov and Marko Petkov{\v{s}}ek.
\newblock Minimal decomposition of indefinite hypergeometric sums.
\newblock In {\em {P}roceedings of {ISSAC}'01}, pages 7--14. ACM, New York,
  2001.

\bibitem{AbPe2002a}
Sergei~A. Abramov and Marko Petkov{\v{s}}ek.
\newblock On the structure of multivariate hypergeometric terms.
\newblock {\em Adv. in Appl. Math.}, 29(3):386--411, 2002.

\bibitem{AbPe2002b}
Sergei~A. Abramov and Marko Petkov{\v{s}}ek.
\newblock Rational normal forms and minimal decompositions of hypergeometric
  terms.
\newblock {\em J. Symbolic Comput.}, 33(5):521--543, 2002.

\bibitem{Andr1976}
George~E. Andrews.
\newblock {\em The {T}heory of {P}artitions}.
\newblock Addison-Wesley Publishing Co., Reading, Mass.-London-Amsterdam, 1976.
\newblock Encyclopedia of Mathematics and its Applications, Vol. 2.

\bibitem{Andr1986}
George~E. Andrews.
\newblock {\em {$q$}-{S}eries: {T}heir {D}evelopment and {A}pplication in
  {A}nalysis, {N}umber {T}heory, {C}ombinatorics, {P}hysics, and {C}omputer
  {A}lgebra}, volume~66 of {\em CBMS Regional Conference Series in
  Mathematics}.
\newblock Published for the Conference Board of the Mathematical Sciences,
  Washington, DC; by the American Mathematical Society, Providence, RI, 1986.

\bibitem{AnAs1999}
George~E. Andrews, Richard Askey, and Ranjan Roy.
\newblock {\em Special Functions}, volume~71 of {\em Encyclopedia of
  Mathematics and its Applications}.
\newblock Cambridge University Press, Cambridge, 1999.

\bibitem{ApZe2006}
Moa Apagodu and Doron Zeilberger.
\newblock Multi-variable {Z}eilberger and {A}lmkvist-{Z}eilberger algorithms
  and the sharpening of {W}ilf-{Z}eilberger theory.
\newblock {\em Adv. in Appl. Math.}, 37(2):139--152, 2006.

\bibitem{BCCL2010}
Alin Bostan, Shaoshi Chen, Fr{\'e}d{\'e}ric Chyzak, and Ziming Li.
\newblock Complexity of creative telescoping for bivariate rational functions.
\newblock In {\em {P}roceedings of {ISSAC}'10}, pages 203--210. ACM, New York,
  2010.

\bibitem{BCCLX2013}
Alin Bostan, Shaoshi Chen, Fr{\'e}d{\'e}ric Chyzak, Ziming Li, and Guoce Xin.
\newblock Hermite reduction and creative telescoping for hyperexponential
  functions.
\newblock In {\em {P}roceedings of {ISSAC}'13}, pages 77--84. ACM, New York,
  2013.

\bibitem{BCLS2018}
Alin Bostan, Fr{\'e}d{\'e}ric Chyzak, Pierre Lairez, and Bruno Salvy.
\newblock Generalized {H}ermite reduction, creative telescoping and definite
  integration of {D}-finite functions.
\newblock In {\em {P}roceedings of {ISSAC}'18}, pages 95--102. ACM, New York,
  2018.

\bibitem{BLS2013}
Alin Bostan, Pierre Lairez, and Bruno Salvy.
\newblock Creative telescoping for rational functions using the
  {G}riffiths-{D}work method.
\newblock In {\em {P}roceedings of {ISSAC}'13}, pages 93--100. ACM, New York,
  2013.

\bibitem{BrSa2024}
Hadrien Brochet and Bruno Salvy.
\newblock Reduction-based creative telescoping for definite summation of
  {D}-finite functions.
\newblock {\em J. Symbolic Comput.}, 125:102329, 2024.

\bibitem{Bron2005}
Manuel Bronstein.
\newblock {\em Symbolic {I}ntegration {I}: Transcendental Functions}, volume~1
  of {\em Algorithms and Computation in Mathematics}.
\newblock Springer-Verlag, Berlin, second edition, 2005.

\bibitem{BLW2005}
Manuel Bronstein, Ziming Li, and Min Wu.
\newblock Picard-{V}essiot extensions for linear functional systems.
\newblock In {\em Proceedings of {ISSAC}'05}, pages 68--75. ACM, New York,
  2005.

\bibitem{Chen2019}
Shaoshi Chen.
\newblock A reduction approach to creative telescoping.
\newblock In {\em Proceedings of {ISSAC}'19}, pages 11--14. ACM, New York,
  2019.

\bibitem{CDK2021}
Shaoshi Chen, Lixin Du, and Manuel Kauers.
\newblock Lazy {H}ermite reduction and creative telescoping for algebraic
  functions.
\newblock In {\em Proceedings of {ISSAC}'21}, pages 75--82. ACM, New York,
  2021.

\bibitem{CDK2023}
Shaoshi Chen, Lixin Du, and Manuel Kauers.
\newblock Hermite reduction for {D}-finite functions via integral bases.
\newblock In {\em Proceedings of {ISSAC}'23}, pages 155--163. ACM, New York,
  2023.

\bibitem{CDKW2025}
Shaoshi Chen, Lixin Du, Manuel Kauers, and Rong-Hua Wang.
\newblock Reduction-based creative telescoping for {P}-recursive sequences via
  integral bases.
\newblock {\em J. Symbolic Comput.}, 126:102341, 2025.

\bibitem{CvHKK2018}
Shaoshi Chen, Mark \Hoeij{van Hoeij}, Manuel Kauers, and Christoph Koutschan.
\newblock Reduction-based creative telescoping for fuchsian {D}-finite
  functions.
\newblock {\em J. Symbolic Comput.}, 85:108--127, 2018.

\bibitem{CHHLW2022}
Shaoshi Chen, Qing-Hu Hou, Hui Huang, George Labahn, and Rong-Hua Wang.
\newblock Constructing minimal telescopers for rational functions in three
  discrete variables.
\newblock {\em Adv. in Appl. Math.}, 141:102389, 2022.

\bibitem{CHKL2015}
Shaoshi Chen, Hui Huang, Manuel Kauers, and Ziming Li.
\newblock A modified {A}bramov-{P}etkov\v sek reduction and creative
  telescoping for hypergeometric terms.
\newblock In {\em Proceedings of {ISSAC}'15}, pages 117--124. ACM, New York,
  2015.

\bibitem{CKK2016}
Shaoshi Chen, Manuel Kauers, and Christoph Koutschan.
\newblock Reduction-based creative telescoping for algebraic functions.
\newblock In {\em Proceedings of {ISSAC}'16}, pages 175--182. ACM, New York,
  2016.

\bibitem{ChSi2014}
Shaoshi Chen and Michael~F. Singer.
\newblock On the summability of bivariate rational functions.
\newblock {\em J. Algebra}, 409:320--343, 2014.

\bibitem{CHM2005}
William Y.~C. Chen, Qing-Hu Hou, and Yan-Ping Mu.
\newblock Applicability of the {$q$}-analogue of {Z}eilberger's algorithm.
\newblock {\em J. Symbolic Comput.}, 39(2):155--170, 2005.

\bibitem{DHL2018}
Hao Du, Hui Huang, and Ziming Li.
\newblock A {$q$}-analogue of the modified {A}bramov-{P}etkov\v sek reduction.
\newblock In {\em Advances in Computer Algebra}, volume 226 of {\em Springer
  Proc. Math. Stat.}, pages 105--129. Springer, Cham, 2018.

\bibitem{GaRa2004}
George Gasper and Mizan Rahman.
\newblock {\em Basic Hypergeometric Series}, volume~96 of {\em Encyclopedia of
  Mathematics and its Applications}.
\newblock Cambridge University Press, Cambridge, second edition, 2004.
\newblock With a foreword by Richard Askey.

\bibitem{GHLZ2019}
Mark Giesbrecht, Hui Huang, George Labahn, and Eugene Zima.
\newblock Efficient integer-linear decomposition of multivariate polynomials.
\newblock In {\em Proceedings of {ISSAC}'19}, pages 171--178. ACM, New York,
  2019.

\bibitem{GHLZ2021}
Mark Giesbrecht, Hui Huang, George Labahn, and Eugene Zima.
\newblock Efficient {$q$}-integer linear decomposition of multivariate
  polynomials.
\newblock {\em J. Symbolic Comput.}, 107:122--144, 2021.

\bibitem{HaSi2008}
Charlotte Hardouin and Michael~F. Singer.
\newblock Differential {G}alois theory of linear difference equations.
\newblock {\em Math. Ann.}, 342(2):333--377, 2008.

\bibitem{vdHo2018}
Joris \Hoeven{van der Hoeven}.
\newblock Creative telescoping using reductions, 2018.
\newblock Preprint: hal-01773137.

\bibitem{vdHo2021}
Joris \Hoeven{van der Hoeven}.
\newblock Constructing reductions for creative telescoping: the general
  differentially finite case.
\newblock {\em Appl. Algebra Engrg. Comm. Comput.}, 32(5):575--602, 2021.

\bibitem{Hou2004}
Qing-Hu Hou.
\newblock {$k$}-free recurrences of double hypergeometric terms.
\newblock {\em Adv. in Appl. Math.}, 32(3):468--484, 2004.

\bibitem{Huan2016}
Hui Huang.
\newblock New bounds for hypergeometric creative telescoping.
\newblock In {\em Proceedings of {ISSAC}'16}, pages 279--286. ACM, New York,
  2016.

\bibitem{Karr1981}
Michael Karr.
\newblock Summation in finite terms.
\newblock {\em J. Assoc. Comput. Mach.}, 28(2):305--350, 1981.

\bibitem{Koor1993}
Tom~H. Koornwinder.
\newblock On {Z}eilberger's algorithm and its {$q$}-analogue.
\newblock {\em J. Comput. Appl. Math.}, 48(1-2):91--111, 1993.

\bibitem{Lang2002}
Serge Lang.
\newblock {\em Algebra}, volume 211 of {\em Graduate Texts in Mathematics}.
\newblock Springer-Verlag, New York, third edition, 2002.

\bibitem{MoZe2005}
Mohamud Mohammed and Doron Zeilberger.
\newblock Sharp upper bounds for the orders of the recurrences output by the
  {Z}eilberger and {$q$}-{Z}eilberger algorithms.
\newblock {\em J. Symbolic Comput.}, 39(2):201--207, 2005.

\bibitem{PaRi1997}
Peter Paule and Axel Riese.
\newblock A {M}athematica {$q$}-analogue of {Z}eilberger's algorithm based on
  an algebraically motivated approach to {$q$}-hypergeometric telescoping.
\newblock In {\em Special functions, {$q$}-series and related topics
  ({T}oronto, {ON}, 1995)}, volume~14 of {\em Fields Inst. Commun.}, pages
  179--210. Amer. Math. Soc., Providence, RI, 1997.

\bibitem{PaSt1995}
Peter Paule and Volker Strehl.
\newblock Symbolic summation - some recent developments.
\newblock In {\em Computer Algebra in Science and Engineering - Algorithms,
  Systems, and Applications}, J. Fleischer, J. Grabmeier, F. Hehl, and W.
  Kuechlin, eds., pages 138--162. World Scientific, Singapore, 1995.

\bibitem{Petk1992}
Marko Petkov{\v{s}}ek.
\newblock Hypergeometric solutions of linear recurrences with polynomial
  coefficients.
\newblock {\em J. Symbolic Comput.}, 14(2-3):243--264, 1992.

\bibitem{PWZ1996}
Marko Petkov{\v{s}}ek, Herbert~S. Wilf, and Doron Zeilberger.
\newblock {\em {$A=B$}}.
\newblock A K Peters, Ltd., Wellesley, MA, 1996.

\bibitem{Stan1995}
Dennis Stanton.
\newblock {\em Talk} at the Workshop “Special Functions, q-Series and Related
  Topics”, organized by the Fields Institute for Research in Mathematical
  Sciences at University College, 12–23 June 1995, Toronto, Ontario.

\bibitem{Wein2009}
Steven~H. Weintraub.
\newblock {\em Galois Theory}.
\newblock Universitext. Springer, New York, second edition, 2009.

\bibitem{Zeil1990a}
Doron Zeilberger.
\newblock A fast algorithm for proving terminating hypergeometric identities.
\newblock {\em Discrete Math.}, 80(2):207--211, 1990.

\bibitem{Zeil1990b}
Doron Zeilberger.
\newblock A holonomic systems approach to special functions identities.
\newblock {\em J. Comput. Appl. Math.}, 32(3):321--368, 1990.

\bibitem{Zeil1991}
Doron Zeilberger.
\newblock The method of creative telescoping.
\newblock {\em J. Symbolic Comput.}, 11(3):195--204, 1991.

\end{thebibliography}
\newcommand{\Gathen}{\relax}\newcommand{\Hoeij}{\relax}\newcommand{\Hoeven}{\relax}\def\cprime{$'$}
  \def\cprime{$'$} \def\cprime{$'$} \def\cprime{$'$} \def\cprime{$'$}
  \def\cprime{$'$} \def\cprime{$'$} \def\cprime{$'$} \def\cprime{$'$}
  \def\polhk#1{\setbox0=\hbox{#1}{\ooalign{\hidewidth
  \lower1.5ex\hbox{`}\hidewidth\crcr\unhbox0}}} \def\cprime{$'$}

\end{document}